\documentclass[12pt,draftclsnofoot,onecolumn]{IEEEtran}

\usepackage{tabulary}
\usepackage{booktabs}

\usepackage{subcaption}
\usepackage{multicol}
\usepackage{setspace}
\usepackage{amsmath}
\usepackage{bbm}
\usepackage[dvips]{color}
\usepackage{comment}
\usepackage{todonotes}
\usepackage{epsf}
\usepackage{epsfig}
\usepackage{times}
\usepackage{epsfig}
\usepackage{graphicx}
\usepackage{bbold}
\usepackage{mathrsfs}
\usepackage{amssymb}
\usepackage{pdfpages}
\usepackage{epstopdf}
\usepackage{tcolorbox,bbold}
\usepackage{authblk}
\usepackage{url}
\usepackage{dsfont}
\usepackage{lettrine} 
\usepackage{amsmath,epsfig,amssymb,algorithm,algpseudocode,amsthm,cite,url}
\IEEEoverridecommandlockouts
\usepackage{subcaption}
\allowdisplaybreaks
\usepackage{csquotes}

\usepackage{verbatim}
\usepackage[english]{babel}
\usepackage{amsmath,amssymb}

\usepackage[justification=centering]{caption}
\usepackage{verbatim}

\newtheorem{theorem}{\bf Theorem}

\begin{document}
\clearpage
\title{\huge Optimized Deployment of Millimeter Wave Networks for In-venue Regions with Stochastic Users' Orientation}
\author{Mehdi Naderi Soorki$^{1}$, Walid Saad$^{1}$, and Mehdi Bennis$^{2}$\vspace*{0.1cm}\\
\small {$^{1}$Wireless@VT, Bradley Department of Electrical and Computer Engineering, Virginia Tech, Blacksburg, VA, USA,Emails: \{mehdin,walids\}@vt.edu.\\
$^{2}$Centre for Wireless Communications, University of Oulu, Finland, Email: bennis@ee.oulu.fi.}
\vspace{-0.5cm}
  \thanks{This research was supported by the U.S. National Science Foundation under Grants CNS-1526844 and IIS-1633363. A preliminary version of this work appears in~\cite{ours2}.}%
}
\maketitle
\thispagestyle{empty}
\vspace{-1cm}
\begin{abstract}
Millimeter wave (mmW) communication is a promising solution for providing high-capacity wireless network access. However, the benefits of mmW are limited by the fact that the channel between a mmW access point and the user equipment can stochastically change due to severe blockage of mmW links by obstacles such as the human body. Thus, one main challenge of mmW network coverage is to enable directional line-of-sight links between access points and mobile devices. In this paper, a novel framework is proposed for optimizing mmW network coverage within hotspots and in-venue regions, while being cognizant of the body blockage of the network's users. In the studied model, the locations of potential access points and users are assumed as predefined parameters while the orientation of the users is assumed to be stochastic. Hence, a joint stochastic access point placement and beam steering problem subjected to stochastic users' body blockage is formulated, under desired network coverage constraints. Then, a greedy algorithm is introduced to find an approximation solution for the joint deployment and assignment problem using a new ``size constrained weighted set cover'' approach. A closed-form expression for the ratio between the optimal solution and approximate one (resulting from the greedy algorithm) is analytically derived. The proposed algorithm is simulated for three in-venue regions: the meeting room in the Alumni Assembly Hall of Virginia Tech, an airport gate, and one side of a stadium football. Simulation results show that, in order to guarantee network coverage for different in-venue regions, the greedy algorithm uses at most three more access points (APs) compared to the optimal solution. The results also show that, due to the use of the additional APs, the greedy algorithm will yield a network coverage up to $11.7\%$ better than the optimal, AP-minimizing solution.
\end{abstract}
\vspace{0.1cm}

\section{Introduction}\label{sec:Intro}
Millimeter wave (mmW) communications over the 30-300 GHz band is a promising approach for overcoming the problem of spectrum scarcity in wireless cellular networks~\cite{Walid6G,shokri2015millimeter,Omid2,baldemair2015ultra}. Due to the large amount of bandwidth available at mmW bands, mmW communication promises to deliver high wireless data rates which makes it an attractive solution for providing wireless connectivity to hotspot regions such as large theaters, arenas, stadiums, shopping malls, and transportation hubs ~\cite{baldemair2015ultra}. In such popular in-venue scenarios, there are tens of thousands of active users packed into a relatively small area~\cite{Kulkarni2017}. Nonetheless, many technical challenges must be overcome to reap the benefits of mmW network deployments and ensure reliable mmW communication~\cite{Bennis}. One prominent challenge is the sensitivity of mmW signals to blockage in dense regions~\cite{singh2015tractable} caused by people, objects in the local environment, and changes in the orientation of the mobile device (MD) carried by the users~\cite{Seong2016GHz}.

In order to overcome propagation challenges such as blockage by humans and buildings, mmW systems typically use beamforming at both access points and mobile devices~\cite{shokri2015millimeter}. Indeed, the use of high-gain directional antennas and dense access point (AP) deployment is necessary for effective mmW communications~\cite{baldemair2015ultra,singh2015tractable} and~\cite{7136141}. A dense AP deployment with high-gain directional antennas can compensate for the significant path loss over mmW frequencies and also allows the establishment of line-of-sight (LoS) links with sufficiently large signal-to-noise ratio. Due to the random blockage of the LoS of the mmW links, optimizing the deployment and beam steering strategies for the APs becomes more challenging than in conventional networks, particularly when deployment is done in three-dimensional space within a venue~\cite{Kovalchukov2018}.

\subsection{Prior works}
Recent works on mmW communications such as~\cite{szyszkowiczautomated,shokri2015user, bai2015coverage,ours,ours2,7136141,8292400,Omid1,Channel_model_value}, and~\cite{gruber2016scalability} have investigated the problems of access point deployment, beam steering, and coverage optimization. In~\cite{7136141}, the authors investigate the effect of the number of antennas on the capacity and coverage probability in the mmW-based small cell, then they propose an antenna clustering scheme to utilize the antennas more efficiently. In addition, they define capacity-maximization and coverage-maximization criteria in the mmW-based small cell. Each design criterion is also formulated as a joint optimization problem in~\cite{7136141}. In~\cite{szyszkowiczautomated}, the authors develop an algorithm that uses computational geometry to place below-rooftop wall-mounted access points using a LoS propagation model for mmW carriers. The goal was to find a set of candidate AP locations whose LoS region, as viewed from each AP in a given set of disjoint geotropical blocks, has a locally maximum area. The authors in~\cite{shokri2015user} proposed a distributed auction-based solution, in which the MDs and APs act asynchronously to achieve optimal MD association and beamforming. In~\cite{shokri2015user}, the problem of jointly optimizing resource allocation and user association in mmW networks is investigated. The goal of this work was to maximize the data rates for the network users while considering load balancing across APs. The authors in~\cite{8292400} study the user-base station association problem in network with the existence of both  mmW and microwave base stations. Considering that each base station has a limited number of resource blocks, an optimization problem is formulated in order to maximize the number of associated users and to ensure an efficient resource utilization by minimizing simultaneously the number of used resource blocks. In~\cite{Omid1,Channel_model_value} the problem of deploying dual-mode base stations that integrate both mmW and microwave frequencies  is investigated.  The authors in~\cite{Omid1} propose a novel framework based on the matching theory to exploit the users' context in resource allocation over the mmW and microwave frequency bands. In~\cite{Channel_model_value}, the problem of cell association is formulated as a one-to-many matching problem with minimum quota constraints for the base stations that provides an efficient way to balance the load over the mmW and microwave frequency bands. To solve the problem, a distributed algorithm is proposed that is guaranteed to yield a Pareto optimal and two-sided stable solution. The authors in~\cite{gruber2016scalability} consider a set of restricted locations for APs and static users. They obtain the optimal number of access points to maximize average user throughput by means of simulations. Then, they show how this optimal number of access points is affected by user distribution and beamwidth of the antennas. Despite treating key challenges of mmW system deployment, the works in~\cite{szyszkowiczautomated,shokri2015user, bai2015coverage,ours,ours2,7136141,8292400,Omid1,Channel_model_value}, and~\cite{gruber2016scalability} consider a simple binary probability model for LoS mmW link, and they do not capture the stochastic blockage of mmW links due to the users' body within real in-venue regions.

In~\cite{bai2015coverage,8292566,8493070,Gapeyenko2016,Geordie2017}, and~\cite{Kovalchukov2018}, the stochastic geometry framework is proposed to evaluate the coverage and rate performance of mmW cellular networks. The authors in~\cite{bai2015coverage} study the effects of mmW blockage by applying a distance-dependent LoS probability function, and modeling the APs as independent inhomogeneous LoS and non-LoS point processes. Then, the mmW coverage and rate performance are analyzed as a function of the antenna geometry and AP density. A K-tier heterogeneous downlink mmW cellular network with user-centric small cell deployments is studied in~\cite{8292566}. In particular, the authors in~\cite{8292566} consider a heterogeneous network model with user equipments being deployed according to a Poisson cluster process. In addition, the MDs are clustered around the base stations and the distances between MDs and the base station are assumed to be Gaussian distributed. Then, using tools from stochastic geometry, they derive a general expression of the signal-to-interference-plus-noise ratio coverage probability. In~\cite{8493070}, the authors consider an open park-like scenario and obtain closed-form expressions for the expected frequency and duration of blockage events using stochastic geometry. Their results indicate that the minimum density of base station, that is required to satisfy the quality of service requirements specially for ultra reliable low latency applications, is largely driven by blockage events rather than capacity requirements. In~\cite{Gapeyenko2016}, the authors propose a tractable model for characterizing the probability of human-body blockage in urban environment. They modeled humans as cylinders with arbitrarily distributed heights and radii, whose centers follow a Poisson point process in two dimensions. By using stochastic geometry, the authors in~\cite{8493070} find the blockage probability as a function of receiver dimension and the transmitter-receiver separation. Then, based on their analysis, the optimal height of the mmW transmitter in crowded outdoor environments is derived and shown to be proportional to the transmitter-receiver separation. In~\cite{Geordie2017}, the authors study the feasibility of mmW frequencies in the wireless wearable devices. They consider a closed indoor scenario where the people are randomly distributed and they derived closed-form expressions for the interference. In~\cite{Kovalchukov2018}, the authors derive the mean interference for emerging 3D mmW communication scenarios where both transmitting and receiving ends have random heights and positions. Although the works in~\cite{Kovalchukov2018}, ~\cite{bai2015coverage}, and~\cite{8292566,8493070,Gapeyenko2016,Geordie2017}, use stochastic geometry to analyze performance, they assume that the users' locations follow a well-defined point process. However, for a given real in-venue region such as a stadium or hall, well-known point processes are not suitable to model the users' locations. This is due to the fact that the locations because in-venue user locations (e.g., in seating charts) are not random. Moreover, in real scenarios, the impact of the stochastic blockage of mmW links due to user blockage can be modeled more accurately given the position of seats for in-venue regions.

More recent works on mmW communications such as~\cite{Kulkarni2017,Dub1}, and~\cite{Rappaport_blockage} have investigated the effects of human body blockage on the availability of LoS mmW links between the access points and mobile devices in real scenarios. The work in~\cite{Dub1} studied dense deployments of millimetre-wave access points with fixed directional antennas mounted on the ceiling. In the setup of~\cite{Dub1}, the main factor limiting signal propagation are blockages by human bodies. They evaluate a number of scenarios that take into account beamwidth of the main-lobe, access point density, and positioning of the mobile device with respect to the user’s body. Then, the authors in~\cite{Dub1} find a trade-off in beamwidth design, as the optimal beamwidth maximizes either coverage or area spectral efficiency, but not both. The work in~\cite{Rappaport_blockage} focuses on blockage events caused by typical pedestrian traffic in a heavily populated open square scenario in Brooklyn, New York. Transition probability rates are determined from the measurements for a two-state Markov and a four-state piecewise linear models. In practice, the availability of LoS mmW links between the access points and mobile devices can be highly dynamic because mmW signals are sensitive to human body blockage and user orientation. In~\cite{Kulkarni2017}, the authors used 3D ray tracing to evaluate the performance of mmW cellular networks in a realistic model of the MetLife stadium. They modeled human blockage features using the dielectric properties at 28 GHz. Then, they showed that meeting a minimum of 100 Mbps rate in large populated venues such as stadiums is very challenging, even in dense mmW networks with high bandwidth. In~\cite{Cisco2011}, a solution called Cisco's Connected Stadium Wi-Fi is designed to provide full coverage throughout venues. However, the bit rate of this solution is limited because the WiFi collision rate will be high in crowded in-venue scenarios.

 Despite treating key challenges of mmW system deployment in a real scenarios, the works in~\cite{Dub1,Rappaport_blockage,Kulkarni2017} and \cite{Cisco2011} completely ignore the impact of the stochastic blockage of mmW links that can result from the randomly changing orientation of the users and their devices in a real in-venue regions.

Recent works such as~\cite{ours2} and~\cite{ours} have considered the stochastic blockage resulting from the orientation of users on the coverage of mmW networks. The work in~\cite{ours} used chance-constrained stochastic programming~\cite{stochasticProgramming} to find the optimal position for APs when the position of the users is given. However, this work does not study the problem of beam steering and it relies on a complex stochastic optimization formulation that cannot be used for a large number of APs and MDs. The work in~\cite{ours2} proposed a greedy algorithm to solve stochastic optimization problems like the one in~\cite{ours} with beam steering. However, the works in~\cite{ours2} and~\cite{ours} do not consider the blockage of nearby users while the human body blockage of nearby users affects on the coverage of the target user specially for crowded in-venue regions.
\subsection{Contributions}
The main contribution of this paper is a novel analytical framework that enables the joint optimization of mmW access point deployment and beam steering while being cognizant of MDs' orientations and blockage of near by users within in-venue mmW networks. In particular, we consider the stochastic blockage of mmW links that is caused by a user's body due to the random orientation of the user devices. The proposed approach explicitly accounts for the three-dimensional nature of the antenna beams of the MDs and APs. We formulate a joint stochastic AP placement and beam steering problem subject to network coverage constraints. In the proposed formulation, given that the connectivity of the mmW links randomly changes due to the stochastic orientation of the users, we minimize the number of required access points and optimize the beam direction of APs to guarantee a required network coverage under the random changes caused by the users' orientation. Since, the complexity of the joint stochastic AP placement and beam steering problem is high specially in three-dimensional space, we propose a new greedy algorithm based on the ``size constrained weighted set cover'' framework~\cite{thomas2001introduction} to find approximate solutions to the joint stochastic AP deployment and MD assignment problem. The closed-form expression between the optimal and approximate solutions is analytically derived. Simulation results show that in order to guarantee coverage constraint, the greedy algorithm uses at most two additional APs in the Alumni Assembly Hall of Virginia Tech and airport gate, and four additional APs in one side of football stadium compared to the optimal solution. Moreover, although the greedy algorithm uses the additional APs compared to the optimal solution, the greedy algorithm will yield a network coverage that is about $3\%$, $11.7\%$, $8\%$ better than the optimal, AP-minimizing solution, for the meeting room in the Alumni Assembly Hall of Virginia Tech, airport gate, and one side of the football stadium, respectively. In summary, the main contributions of this work are:

\begin{itemize}
\item We provide an exact model of the stochastic blockage of mmW links in real-world three-dimensional scenarios in which a large number of users are closely seated within in-venue regions. Our model captures both the blockage of nearby  users as well as the stochastic blockage due to the random orientation of the users.

\item We formulate a new joint stochastic AP placement and beam steering problem for realistic, three-dimensional in-venue regions. Then, we propose a new greedy algorithm and mathematically derive the approximation gap between the optimal and approximate solutions in closed-form.

\item We evaluate the efficiency of our proposed model and algorithm in several realistic settings that include a hall at Virginia Tech, a stadium, and an airport gate. Based on the simulation results, a network operator can practically use our proposed algorithm to place mmW APs. Our proposed algorithm can satisfy the connectivity requirement for in-venue regions including a set of seats located to one another.
\end{itemize}

The rest of the paper is organized as follows. Section~\ref{Sec:Sys-Model} presents the system and the joint stochastic AP placement and beam steering problem while considering the blockage of nearby users. Then, we present our proposed greedy AP placement algorithm for the problem in Section~\ref{Sec:Greedy}. In Section~\ref{Sec:Simulation}, we numerically evaluate the proposed greedy algorithm for different in-venue regions. Finally, conclusions are drawn in Section~\ref{Sec:Conclusion}.
\section{System Model And Problem Formulation}\label{Sec:Sys-Model}
\subsection{System Model}
Consider a set $\mathcal{L}$ of $L$ candidate locations for placing mmW APs in a three-dimensional space. We consider a finite value for $L$ because the maximum possible number of mmW APs should be finite in practice.

Each candidate location $l$ is given by $(x_l,y_l,z_l)$ in Cartesian coordinates. Let $b_l$ be a binary variable which is equal to one if an AP is placed at candidate location $l$, and zero otherwise. We consider fixed directional antennas, where the antennas of mmW AP are pointing down from the ceiling in an in-venue region and creating a spotlight of coverage. In this model, each mmW AP has one wide beam of width $W$ in three-dimensional space~\cite{shokri2015millimeter}. After directing the wide and fixed beam of the mmW AP to a desired direction, based on the resource allocation scheme at the medium access control (MAC) layer, a total of $T$ users can be covered by each beam of mmW AP~\cite{shokri2015millimeter}. The actual antenna pattern is approximated by a flat-top sectored antenna model that is characterized by its pattern function which measures the power gain in polar coordinates around the antenna $g(\theta, \phi)$ over the spherical elevation and azimuthal angle coordinates, $\theta$ and $\phi$~\cite{singh2011interference}. We assume that the spherical elevation and azimuthal angles, $\theta_l$ and $\phi_l$, of the antenna of each AP $l$ are chosen from discrete values in $\Theta=\{\frac{n\pi}{4}|n=0,1,2\}$ and $\Phi=\{\frac{n\pi}{4}|n=0,1,...,7\}$, respectively.

Here, our focus is on in-venue scenarios in which a large number of active users are located within a defined seating chart as is the case in a sports stadium, a lecture hall, a concert venue, a theater, or even an airport scenario where users are located at a gate~\cite{Kulkarni2017}. Clearly, for such scenarios a grid-like model is quite appropriate since seats are pre-defined and located adjacently to one another in a grid-like fashion. Indeed, based on the locations of the seats in such scenarios, we can assume a grid-like structure, which is chosen to match the locations. This assumption becomes even more realistic when the the number of users is high enough to fill all of the seats~\cite{Kulkarni2017}. Thus, based on the structures of the locations of seats, we model the venue as an area that follows a grid-like structure where $\mathcal{M}$ is a set of $M$ grid positions (GPs). In addition, the locations of users within a seating chart can be modeled by a static probabilistic model because the locations of seats are static. The probability of presence of a user at position $m$ is given by $q_m$ and is known a priori to the network operator. This model adequately captures important mmW venues such as hotspots and densely populated areas that include transportation hubs, shopping malls, sport stadiums, and theater halls. We define the transmission gain, $G_{lm}^{\textrm{TX}}$, as the directional gain that AP $l$ adds to the link between AP $l$ and the MD present at GP $m$. Here, $\varphi_{lm}^{\textrm{TX}}$ is the azimuthal angle between the positive x-axis and the direction in which AP $l$ views grid position $m$ in the horizon plane and $\psi_{lm}^{\textrm{TX}}$ is the spherical elevation between the positive z-axis and the direction in which AP $l$ sees GP $m$. Fig.~\ref{Smodel_1} is an illustrative example that shows the  azimuthal angle and spherical elevation as viewed by an AP.
\begin{figure}[!t]
	\begin{center}
		\includegraphics[width=.5\linewidth]{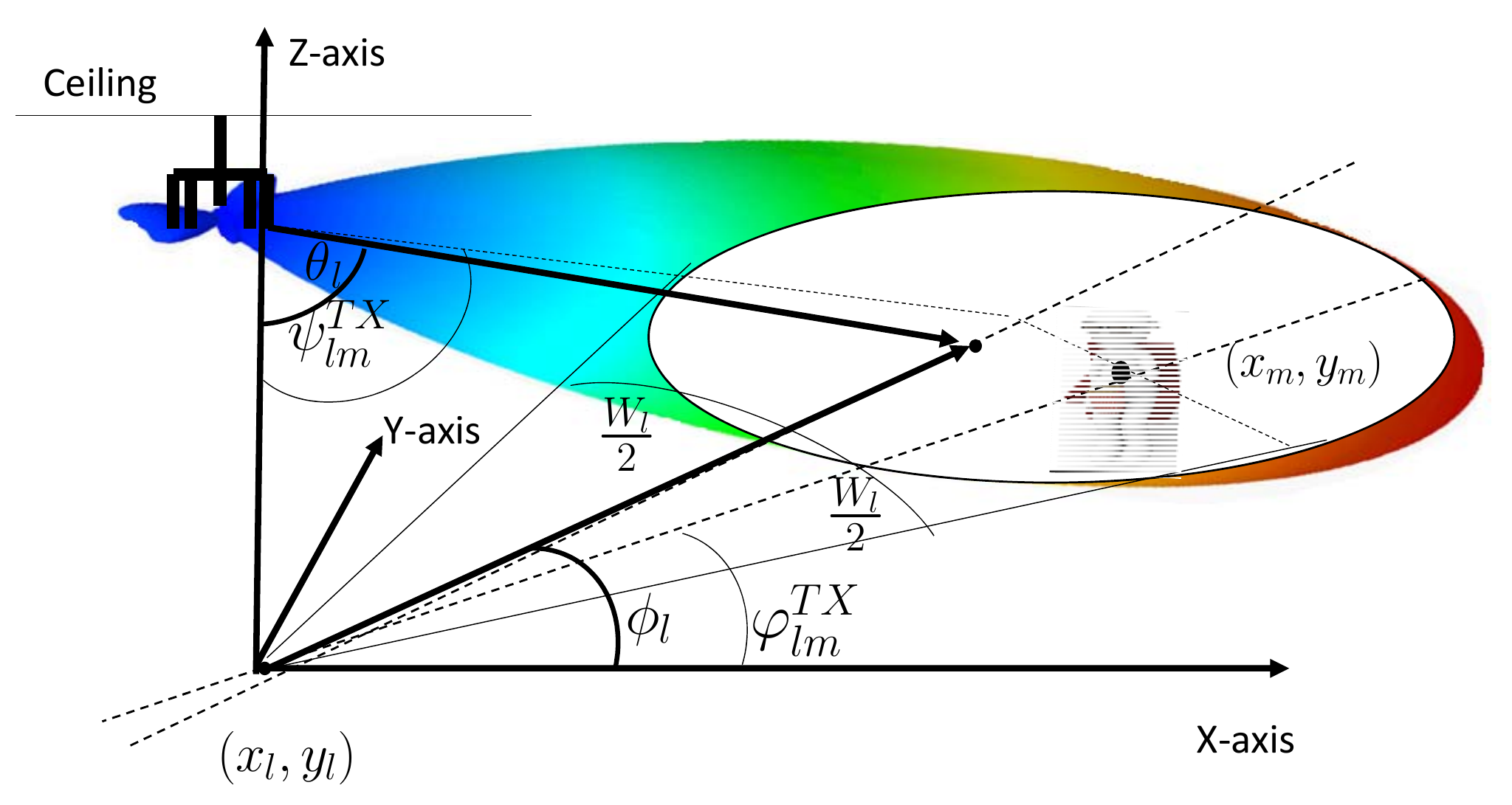}		\vspace{-0.2cm}
		\caption{ \small The azimuthal angle and spherical elevation from the AP view in three-dimensional space.}\vspace{0cm}
		\label{Smodel_1}
	\end{center}
\end{figure}

At each GP, a given user's device can form one narrow beam with width $w$ in the three-dimensional space. Due to the random changes in the orientation of the users within the horizon plane, the azimuthal angle of a given user at GP $m$ is assumed to be a random variable, $\tilde{\phi}_{m}$, with a given probability distribution function $\Pr(\tilde{\phi}_{m}\in \mathcal{B})$, where $\mathcal{B} \subseteq [-\pi,\pi]$. The spherical elevation, $\rho_m$, for any user $m$ at GP $m$, is assumed to be constant. We define the receiver gain $\tilde{G}_{ml}^{\textrm{RX}}$ as the directional gain that an MD located at GP $m$ adds to the link between AP $l$ and GP $m$. $\tilde{G}_{ml}^{\textrm{RX}}$ is a random variable due to the random changes in the orientation of the user.

{In our model, we have accounted for both self-body blockage from a given user as well as blockage from nearby users. The self-body blockage resulting from the body of the user on its own MD is stochastic because the user's orientation randomly changes. On the other hand, the blockage due to the walls and nearby seats in an in-venue region is static. We define $\mathcal{B}_{ml}$ as a set of azimuthal angles within the horizon plane, and $\mathcal{A}_{ml}$ as set of elevation angles within elevation plane that LoS links can be available between a given GP $m$ and AP $l$. Due to the static blockages, these sets, $\mathcal{B}_{ml}$ and $\mathcal{A}_{ml}$, represent in-venue region-dependent variables which are affected by the location of user $m$ and AP $l$, as well as by the locations of other nearby users in densely populated in-venue regions. If $\mathcal{B}_{ml} \neq \emptyset$, $\mathcal{A}_{ml} \neq \emptyset$, and $\phi_{m} \in \mathcal{B}_{ml}$ and $\rho_{m} \in \mathcal{A}_{ml}$, a mmW LoS link is not statically blocked between $m$ and AP $l$, however this LoS link can be stochastically blocked due to the random changes in the user's orientation}.

Consequently, based on the large-scale channel effects over the mmW links following the popular model of~\cite{Channel_model} and the availability of mmW LoS link in our model, the channel gain in dB for mmW link between GP $m$ and AP $l$ is given by:
\begin{equation}
\tilde{h}_{ml}=
  \begin{cases}
      -\kappa-\alpha_L10\log_{10}d_{ml}-\chi_L,& \quad \text{if } \tilde{\phi}_{m} \in \mathcal{B}_{ml},{\rho}_{m} \in \mathcal{A}_{ml},\\
    -\kappa-\alpha_N10\log_{10}d_{ml}-\chi_N, & \quad \text{else},\\
  \end{cases}
\label{Path_loss}
\end{equation}
where $\kappa$ is the path loss (in dB) for 1 meter of distance, $\alpha_L$ and $\alpha_N$ respectively represent the slopes of the best linear fit to the propagation measurement in mmW frequency band for LoS and non-LoS mmW links. In addition, $\chi_L$ and $\chi_N$ model the deviation in fitting (in dB) for LoS and non-LoS mmW links, respectively.  $\chi_L$ and $\chi_N$ are Gaussian random variables with zero mean and variance $\varepsilon_L^2$ and $\varepsilon_N^2$. $d_{ml}$ is the distance between GP $m$ and AP $l$. $\tilde{h}_{ml}$ is a random variable due to the random changes in the orientation of the user and also the blockage of nearby users in in-venue region.

Fig.~\ref{Smodel_2} is an illustrative example that shows the  azimuthal angle and  a LoS-angle set for one GP and two APs. As we can see in Fig.~\ref{Smodel_2}, the LoS-angle sets $\mathcal{B}_{m1}\neq \emptyset$ and $\mathcal{A}_{m1}\neq \emptyset$. Thus, AP $1$ and GP $m$ can have a LoS link because there is no blockage in azimuth and elevation angles. Since $\mathcal{B}_{m2} \neq\emptyset$  and $\mathcal{A}_{m2}=\emptyset$, there is a non-LoS mmW link between AP $2$ and GP $m$ because there is a blockage in the elevation angle.
\begin{figure}[!t]
	\begin{center}
		\includegraphics[width=.5\linewidth]{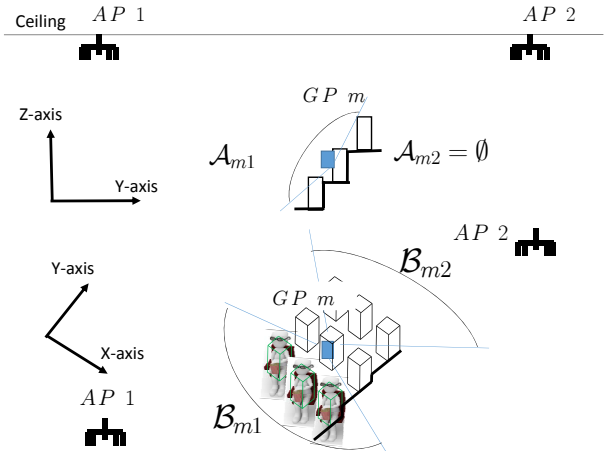}		\vspace{-0.2cm}
		\caption{ \small The LoS-angle sets for the  azimuthal and elevation angles.}\vspace{0cm}
		\label{Smodel_2}
	\end{center}
\end{figure}

Since the orientations of the users can change randomly, the total directional gain between serving APs and users, $G_{lm}^{\textrm{TX}} \times \tilde{G}_{ml}^{\textrm{RX}}$ is a random variable. The AP placement problem is trivial if an AP located at one candidate location can guarantee the availability of LoS links for the users. However, due to the random changes in user orientations, the users may not have an available LoS link even if the beam of the AP can be steered. One solution to overcome this challenge can be to assign more than one AP to each user in order to guarantee coverage for all users, under stochastic orientations, as we formulate next.

\subsection{Joint stochastic AP placement and beam steering problem}\label{Sec:Optimum}
Following the mmW AP placement framework, when a given user's body blocks the LoS mmW link of one mmW AP because of a change of orientation, the beams of more than one mmW APs must be steered toward that user. However, the interfering signals from other mmW APs can be cancelled in the spatial domain by the user's body blockage, particularly, for in-venue regions. Hence, the signal-to-noise ratio (SNR) at GP $m$ resulting from the transmission of an AP at candidate location $l$ is $\tilde{\gamma}_{lm}=\frac{p \tilde{G}_{ml}^{\textrm{RX}} G_{lm}^{\textrm{TX}} \tilde{h}_{ml}}{\sigma^2}$, where $p$ is the transmission power of AP. Given the large amount of bandwidth available at millimeter wave frequencies and the relatively short propagation distance, we assume that multi-user interference is managed via a suitable multiple-access scheme and as a result is negligible~\cite{singh2011interference}. Since the user orientation is stochastic, the SNR over the mmW links for each GP will not be deterministic. Let $\mathcal{L}_m$ be the set of APs where ${\gamma}_{lm}^* \leq {\gamma}_{lm}$. Here, $\gamma_{lm}^*$ represents the minimum SNR requirement needed to have an active communication link between a user at GP $m$ and an AP at candidate location $l$.

According to the orientation, $\tilde{\phi}_{m}$, of MD $m$, $\tilde {\gamma}_{lm}$ will be a random variable for each GP $m$ and AP in $\mathcal{L}_m$. To ensure reliable mmW communication under the random changes caused by the users' orientation, millimeter APs should have enough connections to the user to transmit a given amount of traffic with high success probability~\cite{Bennis}. Thus, the \textit{user connectivity} constraint $\beta_m$ for user $m$ is defined using a pre-determined threshold for the probability that ${\gamma}_{m}^*~ < \tilde{\gamma}_{m}$. The user connectivity constraint $\beta_m$ for user $m$ is given by:
\begin{align}
& \Pr\left\{\sum_{l\in \mathcal{L}_m} \tilde{y}_{lm} a_{lm}\geq 1\right\}\geq \beta_m,
\label{Stochastic_C}
\end{align}
where $a_{lm}$ is a binary variable that equals to one if GP $m$ is assigned to AP $l$, and $\tilde{y}_{lm}$ is a binary random variable. The probability that $\tilde{y}_{lm}$ equals to one depends on the orientation of user $m$ and location of AP $l$.

The chance-constrained method is a formulation of an optimization problem that ensures that the probability of meeting a certain constraint is above a certain level~\cite{stochasticProgramming}. In other words, the chance-constrained method restricts the feasible set of solution for the stochastic optimization problem so that the confidence level of the solution becomes high enough~\cite{stochasticProgramming}. Since the orientation of the users stochastically changes, the mmW link between a user and its access point may be randomly blocked by users' bodies. Thus, as one of the major approaches to guarantee user connectivity constraint in (\ref{Stochastic_C}), we use chance-constrained method under the various uncertainties in the orientation of the users. Let $\omega \in \Omega$ be the index of any given scenario for ${y}_{lm}^{(\omega)}$. The total number of scenarios for AP assignment per user is $2^L$. According to the orientation $\phi_m$ of MD $m$, each scenario $\omega$ has a probability of $q_m^{(\omega)}$ for each MD $m$. For a given user, the probability of each scenario is related to that user's orientation as well as the locations of the other users the in-venue region. This probability is given by $q_m^{(\omega)}=\Pr\left\{\phi_m \in \cap_{l:{y}_{lm}^{(\omega)}=1} \mathcal{B}_{ml},
\rho_m \in \cap_{l:{y}_{lm}^{(\omega)}=1} \mathcal{A}_{ml} \right\}$, where $\mathcal{B}_{ml}=[B_{ml,1},B_{ml,2}]$ and $\mathcal{A}_{ml}=[A_{ml,1},A_{ml,2}]$. Here, $B_{ml,1}$ and $A_{ml,1}$ are the lower bounds and $B_{ml,2}$ and $A_{ml,2}$ are the upper bounds for $\mathcal{B}_{ml}$ and $\mathcal{A}_{ml}$, respectively. Then, following chance-constrained stochastic programming, we guarantee that the coverage constraint of each user is satisfied for a predefined number of scenarios. Thus, the constraint in (\ref{Stochastic_C}) can be equivalently represented by an auxiliary variable $u_m^{(\omega)}$, where $u_m^{(\omega)}$ is a new binary decision variable. $u_m^{(\omega)}$ equals one if under scenario $\omega$ the coverage demand of MD $m$ is not satisfied, otherwise $u_m^{(\omega)}$ equals zero. If $u_m^{(\omega)}=\mathds{1}_\{  \gamma_{m}^{(\omega)} <  {\gamma}_{m}^* \}$, then constraint (\ref{Stochastic_C}) can be given by:
\begin{align}
& \sum_{l\in \mathcal{L}} {y}_{lm}^{(\omega)} a_{lm}\geq (1-u_{m}^{(\omega)}) \; \forall m \in \mathcal{M},\forall \omega \in \Omega, \label{Lin_Stochastic_C1}\\
& \sum_{\omega \in \Omega} q_m^{(\omega)} u_m^{(\omega)}\leq 1-\beta_m, \forall m \in \mathcal{M}.
\label{Lin_Stochastic_C2}
\end{align}

We define the \textit{network coverage} as the summation of the presence probability of the users whose connectivity constraints are satisfied:
\begin{equation}
\sum_{m\in \mathcal{M}} q_m \mathds{1}_{ \big\{ \sum_{\omega \in \Omega} q_m^{(\omega)} (1-u_m^{(\omega)}) \geq \beta_m \big\} } \geq \alpha.
\label{network_coverage}
\end{equation}

The network coverage constraint in (\ref{network_coverage}) guarantees that the sum of the probabilities $q_m$ associated with those locations where the probability of meeting the SNR threshold is at least $\beta_m$ is at least $\alpha$. Here, $\alpha$ is a constant value between zero and one. (\ref{network_coverage}) is nonlinear and, thus, it can be equivalently represented by an auxiliary binary variable $z_l$, where $z_m=1$ if the APs that are assigned to GP $m$, can guarantee the user connectivity requirement of the user at GP $m$, and $z_m=0$ otherwise. Let $z_m=\mathds{1}_{ \big\{ \sum_{\omega \in \Omega} q_m^{(\omega)} (1-u_m^{(\omega)}) \geq \beta_m \big\} }$. Consequently, the joint AP placement and beam steering problem can be formulated as the following stochastic optimization problem:
\begin{align}
& \underset{\left\{\substack{b_l,\phi_l,\theta_l,a_{lm},y_{lm}^{(\omega)},u_{m}^{(\omega)},z_m\\ l \in \mathcal{L},m \in \mathcal{M},
\omega \in \Omega}\right\}}{\min} \; \sum_{l\in \mathcal{L}}  \; b_l, \label{optprob_Stoc_C}\\
& \hspace{0.2in} \text{s.t.}    \nonumber \\
& \hspace{0.2in}\sum_{l\in \mathcal{L}} {y}_{lm}^{(\omega)} a_{lm}\geq (1-u_{m}^{(\omega)}) \; \forall m \in \mathcal{M},\forall \omega \in \Omega, \label{Lin_Stochastic_C1}\\
& \hspace{0.2in} \sum_{\omega \in \Omega} q_m^{(\omega)} u_m^{(\omega)}\leq 1-\beta_m, \forall m \in \mathcal{M},\label{Lin_Stochastic_C1}\\
& \hspace{0.2in} \sum_{\omega \in \Omega} q_m^{(\omega)} (1-u_m^{(\omega)}) \geq z_m \beta_m, \forall m \in \mathcal{M}, \label{optprobc_Stoc_C2} \\
& \hspace{0.2in} 1-\sum_{\omega \in \Omega} q_m^{(\omega)} (1-u_m^{(\omega)}) \geq (1-z_m)(1-\beta_m), \forall m \in \mathcal{M}, \label{optprobc_Stoc_C2} \\
& \hspace{0.2in} \sum_{m\in \mathcal{M}} q_m z_m \geq \alpha,\\
& \hspace{0.2in} \sum_{l\in \mathcal{L}} b_l \leq L \\
& \hspace{0.2in} b_l \leq \sum_{m\in \mathcal{M}}a_{lm}\leq Tb_l, \forall l \in \mathcal{L}, \label{optprobc_Stoc_C3}\\
&\hspace{0.2in} -\frac{W}{2}\leq a_{lm}(\phi_l-\phi_{lm}^{\textrm{TX}})\leq \frac{W}{2},
\forall m \in \mathcal{M},\forall l \in \mathcal{L}_m,\\
&\hspace{0.2in} -\frac{W}{2}\leq a_{lm}(\theta_{l}-\psi_{lm}^{\textrm{TX}}) \leq \frac{W}{2},
\forall m \in \mathcal{M},\forall l \in \mathcal{L}_m,\\
&\hspace{0.2in}  b_{l},a_{lm} \in \{0,1\},\forall l \in \mathcal{L},\forall m \in \mathcal{M},  \label{optprobc_Stoc_C4}\\
& \hspace{0.2in} \theta_l \in \Theta,\phi_l \in \Phi,\forall l \in \mathcal{L},\label{optprobc_Stoc_C6}\\
& \hspace{0.2in} y_{lm}^{(\omega)},z_{m},u_{m}^{(\omega)} \in \{0,1\}, \forall l \in \mathcal{L},\forall m \in \mathcal{M},\forall \omega \in \Omega. \label{optprobc_Stoc_C5}
\end{align}

The complexity of (\ref{optprob_Stoc_C}) is $O(2^L|\Theta||\Phi|)$ if one uses an exhaustive search algorithm. Hence, it is infeasible to use a brute force algorithm for solving the dense mmW AP placement and beam steering problem specially for scenarios with large size. Here, we note that a standard optimizer such as CPLEX~\cite{Cplex} can be used to solve this problem with a faster computational speed compared to the exhaustive search. However, the computation is still time-consuming when the number of candidate locations increases. To overcome this complexity challenge, next, we propose a novel and efficient greedy algorithm using notions from ``set covering''~\cite{thomas2001introduction}, and~\cite{Set_cover}.

In our model, we consider a sectored antenna model for each AP. The transmission gains are assumed to be equal to a constant value $\frac{2}{1-\cos (\frac{W}{2})}$ for angles in the main lobe, and $g$ for angles in the side lobe~\cite{bai2015coverage}. Thus, the transmission gain of AP $l$ to GP $m$ is given by $G_{lm}^{\textrm{Tx}}=\frac{2}{1-\cos (\frac{W}{2})}$ if $-\frac{W}{2}\leq \phi_{l}-\varphi_{lm}^{\textrm{Tx}}\leq \frac{W}{2}$  and $-\frac{W}{2}\leq \theta_{l}-\psi_{lm}^{\textrm{Tx}}\leq\frac{W}{2}$, otherwise $G_{lm}^{\textrm{Tx}}=g$. Moreover, for each MD, the receive gains are assumed to be equal to a constant value $\frac{2}{1-\cos (\frac{w}{2})}$ for angles in the main lobe and $g$ for angles in the side lobe\cite{bai2015coverage}. Thus, the receive antenna gain of GP $m$ from AP $l$ is given by $\tilde{G}_{ml}^{\textrm{Rx}}=\frac{2}{1-\cos (\frac{w}{2})}$, if $-\frac{w}{2}\leq \tilde{\phi}_{m}-\varphi_{ml}^{\textrm{Rx}}\leq \frac{w}{2}$  and $-\frac{w}{2}\leq\rho_{m}-\psi_{ml}^{\textrm{Rx}}\leq \frac{w}{2}$, otherwise $\tilde{G}_{ml}^{\textrm{Rx}}=g$. Here, $\varphi_{ml}^{\textrm{RX}}$ is the azimuthal angle between the positive x-axis and the direction in which GP $m$ sees AP $l$ and $\psi_{ml}^{\textrm{RX}}$ is the spherical elevation between the positive z-axis and the direction in which GP $m$ sees AP $l$. The list of main notations used throughout this paper is presented in Table~\ref{tab:Symbols}.

\begin{table}[ht]
	\caption{{List of main notations used throughout the paper.}}
	\begin{center}
		\begin{tabular}{l | l } \toprule
{ {\textbf{Symbol}}}  & { {\textbf{Definition}}}   \\  \hline
{$\mathcal{L}$}   &  {Finite set of $L$ candidate locations for APs.}\\
{$(x_l,y_l,z_l)$} &  {Cartesian coordinate representing the location of mmW AP $l$.}\\
{$W$ } & {Beamwidth of AP.}\\
{$T$} &  {Maximum number of users per beam of each AP.}\\
{$\theta_l$} &  {Spherical elevation angle of the antenna of AP $l$.} \\
{$\phi_l$} &  {Azimuthal angle of the antenna of AP $l$.}\\
{$\mathcal{M}$} & { Finite set of $M$ grid positions.}\\
{$q_m$ } &  {The probability of presence of a user at position $m$.} \\
{$G_{lm}^{\textrm{TX}}$} &  {Antenna gain of AP $l$ over the link between AP $l$ and GP $m$.}\\
{$\varphi_{lm}^{\textrm{TX}}$} &  {Azimuthal angle between the positive x-axis and the direction.}\\
 $$ & {in which AP $l$ views grid position $m$.}\\
{$\psi_{lm}^{\textrm{TX}}$ }& {Spherical elevation angle between the positive z-axis and the direction}\\
$$ & {in which AP $l$ sees GP $m$.}\\
{$w$} & { Narrow beam width of each user's device.}\\
{$\tilde{\phi}_{m}\in \mathcal{B}$} & {Azimuthal angle of a given user at GP $m$.}\\
{$\rho_m$} &   {Spherical elevation angle for a user at GP $m$.}\\
{$\mathcal{B}_{ml}$} & {A set of azimuthal angles that LoS links are available between a given GP $m$ and AP $l$.}\\
{$\mathcal{A}_{ml}$}& {A set of elevation angles that LoS links are available between a given GP $m$ and AP $l$.}\\
{$\tilde{h}_{ml}$} & {Channel gain for mmW link between GP $m$ and AP $l$.}\\
{$\tilde{\gamma}_{lm}$} &  {SNR at GP $m$ resulting from the transmission of an AP at candidate location $l$.}\\
{$\gamma_{lm}^*$ } & {Minimum requirement for SNR between user at GP $m$ and AP $l$.}\\
{$B_{ml,1}$} &  {Lower band for $\mathcal{B}_{ml}$.} \\
{$B_{ml,2}$} &  {Higher band for $\mathcal{B}_{ml}$.} \\
{$A_{ml,1}$} & {Lower band for $\mathcal{A}_{ml}$.}\\
{$A_{ml,2}$}& { Higher band for $\mathcal{A}_{ml}$.}\\
{$\alpha$} & {Network coverage constraint.}\\
{$\beta$}&  {User connectivity requirement.}\\
{$L^*$}&  {Number of APs based on optimal solution.}\\
{$L^\circ$}& {Number of APs based on greedy solution.}\\
\hline		
		\end{tabular}
	\end{center}
	\label{tab:Symbols}
\end{table}

\section{Greedy Algorithm for the AP Deployment and Beam Steering Problem}\label{Sec:Greedy}

\begin{table}[!t]
  \centering
  \caption{
    \vspace*{-0em}The proposed greedy algorithm for AP deployment and beam steering in mmW networks.}\vspace*{-0cm}
    \begin{tabular}{p{3.3in}}
      \hline \vspace*{-0em}
      \textbf{Inputs:}\,\,Candidate location, $\mathcal{L}$,\\
       Grid position, $\mathcal{M}$,\\
       Orientation probability for each GP, $\Pr(\phi_m)$,\\
       Connectivity requirement for each GP, $\beta_m$,\\
       Network coverage requirement, $\alpha$.\\
       \textbf{1:}\,\,{Assign initial empty set to the selected AP set}, $\mathcal{L}^\circ= \emptyset$.\\
       \textbf{2:}\,\,{Assign initial empty set to the covered GP set}, $\mathcal{M}^\circ= \emptyset$.\\
       \textbf{3:}\,\,While $\sum_{m\in \mathcal{M}} q_m z_m < \alpha \times M$.\\
       \textbf{ }\,\, Select {AP candidate location $i \in \mathcal{L}$, elevation angle $\theta_i$, and azimuthal angle $\phi_i$ }that maximize
       $\sum_{m\in \mathcal{M}\backslash \mathcal{M}^\circ} z_m $.\\
       \textbf{ }\,\, Set $\mathcal{L}=\mathcal{L} \backslash \{i\}$.\\
       \textbf{ }\,\, Set $\mathcal{L}^\circ=\{i\} \cup \mathcal{L}^\circ$.\\
       \textbf{ }\,\, Set $\mathcal{M}^\circ=\cup_{l \in \mathcal{L}^\circ} \mathcal{C}_l$.\\
       \textbf{4:}\,\,Return $\mathcal{L}^\circ$\\
\textbf{Output:}\,\,AP placement set, $\mathcal{L}^\circ$, AP beam steering, and AP assignment to the GP.\vspace*{0em}\\
   \hline
    \end{tabular}\label{SC_A}\vspace{-0.7cm}
\end{table}

Let $\mathcal{F}$ be a family of subsets of $\mathcal{M}$. Every element of $\mathcal{F}$, $\mathcal{C}_l\in \mathcal{F}$,  corresponds to the set of GPs that can be covered if an AP is placed at candidate location $l\in\mathcal{L}$, and its  beam is steered to a direction having $\theta_l$ and $\phi_l$. Suppose that $\mathcal{C}_l$ covers $C_l=\sum_{m\in \mathcal{C}_l}z_m$ GPs. For a given set $\mathcal{L}$ of selected candidate AP locations, the weight of set $\mathcal{C}_l$ is equal to the sum of the users in $\mathcal{C}_l$ whose connectivity requirement is guaranteed.

The problem of placing the least number of APs under network coverage and user connectivity constraints in mmW networks can be seen as a special case of a well-known ``size-constrained weighted set cover'' problem~\cite{thomas2001introduction}. If we pose our problem in (\ref{optprob_Stoc_C}) as a size-constrained weighted set cover problem, then the input will be a set of $M$ GPs, a collection of weighted sets over the GPs, $\mathcal{C}_l$, a size constraint $L$, and a minimum coverage requirement $\alpha$. The output will be a sub-collection of up to $L$ subsets of grid points that have a maximum sum of weights. Then, we propose a greedy algorithm which approximates the size-constrained weighted set cover problem. In the proposed greedy algorithm, we start with a given AP $l$ that covers $\mathcal{C}_l$ with a high weight value which is likely going to be insufficient to cover the desired number of grid point. Then, we iteratively add more APs with highest marginal benefit to guarantee the required coverage constraint in (\ref{network_coverage}). Table~\ref{SC_A} shows the proposed greedy algorithm. The input parameters of the proposed greedy algorithm in Table~\ref{SC_A} are: the set of candidate locations, $\mathcal{L}$, the set of GPs, $\mathcal{M}$, the orientation probability for a user in each GP, $\Pr(\phi_m)$, and the coverage threshold $\alpha$. The objective of line 2-1 in Table~\ref{SC_A} is to select a candidate AP that covers the set of GPs with the highest marginal benefit. Let $\mathcal{L}^\circ$ be the set of candidate locations that are already selected by iteratively greedy algorithm. At each iteration $i$ of the greedy algorithm, an AP is selected as follows:
\begin{align}
& \underset{\left\{\substack{\phi_i,\theta_i,a_{im},y_{im}^{(\omega)},u_{m}^{(\omega)},z_m\\ i \in \mathcal{L},m \in \mathcal{M},
\omega \in \Omega}\right\}}
{\max} \;
\sum_{m\in \mathcal{M}} q_mz_m, \label{Greedprob_Stoc_C}\\
& \hspace{0.2in} \text{s.t.} \nonumber \\
& \hspace{0.2in} \sum_{l\in\{i\}\cup \mathcal{L}^\circ} {y}_{lm}^{(\omega)} a_{lm}\geq (1-u_{m}^{(\omega)}) \; \forall m \in \mathcal{M},\forall \omega \in \Omega, \label{Lin_Stochastic_C1}\\
& \hspace{0.2in} \sum_{\omega \in \Omega} q_m^{(\omega)} u_m^{(\omega)}\leq 1-\beta_m, \forall m \in \mathcal{M},\label{Lin_Stochastic_C1}\\
& \hspace{0.2in} \sum_{\omega \in \Omega} q_m^{(\omega)} (1-u_m^{(\omega)}) \geq z_m \beta_m, \forall m \in \mathcal{M}, \label{optprobc_Stoc_C2} \\
& \hspace{0.2in} 1-\sum_{\omega \in \Omega} q_m^{(\omega)} (1-u_m^{(\omega)}) \geq (1-z_m)(1-\beta_m), \forall m \in \mathcal{M}, \label{optprobc_Stoc_C2} \\
& \hspace{0.2in} -\frac{W}{2}\leq a_{im}(\phi_i-\phi_{im}^{\textrm{TX}})\leq \frac{W}{2},\forall m \in \mathcal{M},\\
& \hspace{0.2in} -\frac{W}{2}\leq a_{im}(\theta_{i}-\psi_{im}^{\textrm{TX}}) \leq \frac{W}{2},\forall m \in \mathcal{M},\\
& \hspace{0.2in} 0 \leq \sum_{m\in \mathcal{M}}a_{im}\leq T, \label{Greedprobc_Stoc_C3}\\
& \hspace{0.2in} a_{im}\in \{0,1\},\forall m \in \mathcal{M}, \label{Greedprobc_Stoc_C5}\\
& \hspace{0.2in} \theta_i\in \Theta ,\phi_i \in \Phi,\label{optprobc_Stoc_C6}\\
& \hspace{0.2in} y_{lm}^{(\omega)},z_m,u_{m}^{(\omega)} \in \{0,1\},\forall l \in\{i\}\cup \mathcal{L}^\circ,\forall m \in \mathcal{M},\forall \omega \in \Omega \label{optprobc_Stoc_C5}
\end{align}

At each iteration $i$ of the proposed algorithm in Table~\ref{SC_A}, the complexity for finding the best AP location and steering its beam is $(L-i)|\Theta| |\Phi|$. The maximum number of iterations of the proposed greedy algorithm in Table~\ref{SC_A} is $L$. Thus, the complexity of the proposed greedy algorithm in Table~\ref{SC_A} is $O(L^2|\Theta||\Phi|)$ which is proportional to the square of the number of access points $L$. Compared to the exponentially growing complexity of exhaustive search for (\ref{optprob_Stoc_C}), the complexity of proposed greedy algorithm is clearly more reasonable. Next, we compute the approximation gap between the proposed greedy algorithm solution in Table~\ref{SC_A} and the optimal solution in (\ref{optprob_Stoc_C}). We define $\mathcal{L}^\circ$ as the set of APs resulting from the proposed algorithm, and the set $\mathcal{C}_i^\circ$ with size $C_i^\circ$ as the set of GPs that are covered by AP $i\in\mathcal{L}^\circ$. Thus, the set of all GPs covered by proposed algorithm is $\mathcal{M}^\circ=\cup_{i\in\mathcal{L}^\circ}\mathcal{C}_i^\circ$. Let $\mathcal{L}^*$ be the set of APs that are found for the optimal solution, and the set $\mathcal{C}_i^*$ with size $C_i^*$ be the set of GPs that are covered by AP $i\in\mathcal{L}^*$. Thus, the set of all GPs covered by optimal solution is $\mathcal{M}^*=\cup_{i\in\mathcal{L}^*}\mathcal{C}_i^*$. Thus, we can state the following theorem for our proposed algorithm in Table~\ref{SC_A}.

\begin{theorem}
\textnormal{The proposed greedy algorithm returns a solution with up to $\frac{\max_{i} C_i^* \times \max_{m \in \mathcal{M}^*} {q_m} }{\min_{i} C_i \times \min_{m \in \mathcal{M}^*} {q_m}}L^*$ APs to cover the same set of GPs that the optimal solution covers.}
\end{theorem}
\begin{proof}
\textnormal{{See the Appendix A.}}
\end{proof}

The result of Theorem 1 means that the ratio of number of APs selected by greedy algorithm to the optimal solution becomes less when the ration of the maximum number of GPs per AP from optimal solution to the minimum number of GPs per AP from the greedy algorithm becomes less. This ratio depends on how the beamwidths of the APs from the optimal solution and greedy algorithm are selected.

\section{Simulation Results and Analysis}\label{Sec:Simulation}
For our simulations, we consider three in-venue scenarios: the meeting room of Alumni Assembly Hall of Virginia Tech~\cite{Assembly_Hall_ref}, an airport gate, and one side of a stadium football. For this setting, the main-lobe and side-lobe antenna gains are set to $18$~dB and $-2$~dB, respectively~\cite{singh2015tractable}. The path loss $\kappa$ for 1 meter of distance is 70dB, path loss exponents for LoS and non-LoS mmW links, $\alpha_L$ and $\alpha_N$, are 2 and 4, and the standard deviations of path losses over LoS and non-LoS mmW links, $\varepsilon_L$ and $\varepsilon_N$, are 5.2 and 7.6~\cite{Channel_model_value}.We assume that the orientation of an MD at each GP is a random variable between $-\pi$ and $\pi$ that follows a truncated Gaussian distribution whose mean value is the azimuthal angle toward the direction of the seat in the horizon plane. We assume that the beamwidth of MD's antenna is $\frac{\pi}{2}$.
\subsection{Alumni Assembly Hall}
The meeting room of Alumni Assembly Hall of Virginia Tech has 135 seats as shown in Fig.~\ref{Alumni_Assembly_center}. The total area of the meeting space is 1000 sq. meters with the height of the ceiling ranging from 3.40 meter to 4.37 meter. The number of candidate APs which are on a grid-like structure on the ceiling is 20. The direction of GPs is toward the center of the front stage in the horizon plane. Fig.~\ref{Alumni_Assembly_center} shows an illustrative example of AP placement and beam steering resulting from both the optimal solution and the greedy algorithm.

The example in Fig.~\ref{Alumni_Assembly_center} is for an AP beamwidth of $\frac{2\pi}{3}$, an MD beamwidth of $\frac{\pi}{3}$, a network coverage of $\alpha=0.75$, and a user connectivity of $\beta=0.95$. Since most of the time the users' orientation is turned toward the center of the front stage, the APs of both the optimal solution and the greedy algorithm are placed at the front of the Assembly Hall (See: Fig.~\ref{Alumni_Assembly_center}).
\begin{figure}[t!]
  \centering
  \includegraphics[width=.5\linewidth]{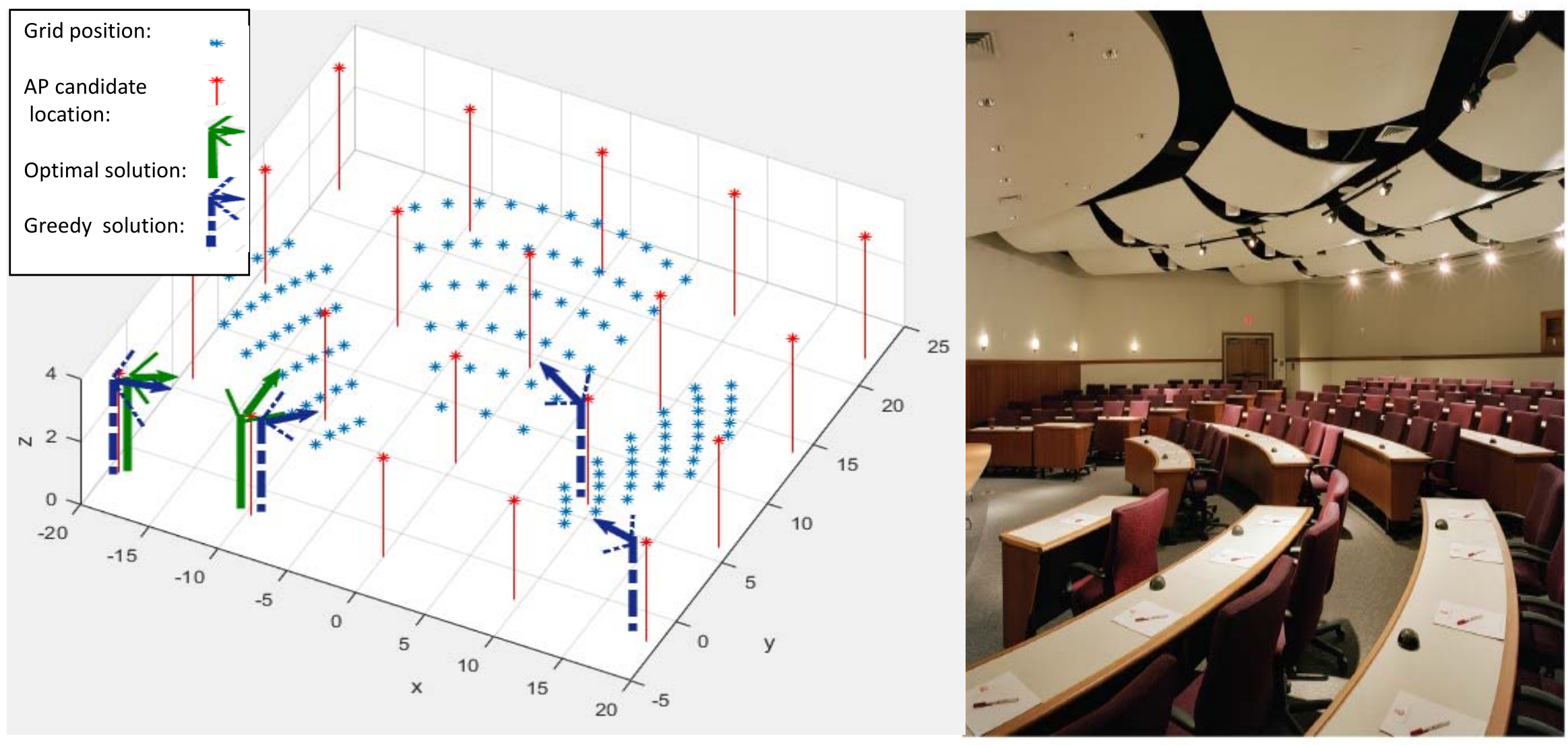}
  \caption{AP deployment and beam steering for the meeting room of Alumni Assembly Hall of Virginia Tech when $W=\frac{10\pi}{15}$, $w=\frac{\pi}{2}$, $\alpha=0.9$, $\beta=0.9$.}
  \label{Alumni_Assembly_center}
\end{figure}%

In Fig.~\ref{N_AP_Hall}, we show the number of required APs for different network coverage and user connectivity constraints versus the AP beamwidth. Fig.~\ref{N_AP_Hall} shows that, when the network coverage and user connectivity constraints increase, the number of required APs also increases. Moreover, for a network coverage $\alpha=0.95$ and user connectivity constraint of $0.9$, the greedy algorithm uses at most two additional AP compared to the optimal solution, while the greedy algorithm uses one additional AP compared to the optimal solution for $\alpha=0.65$ and $\beta=0.7$. {Moreover, the uniform solution uses more APs compared to the optimal and greedy solutions. Fore example, the uniform solution uses 16 APs for $\alpha=0.95$ and $\beta=0.9$ to guarantee the network coverage constraint.}
\begin{figure}[t!]
  \centering
  \includegraphics[width=.5\linewidth]{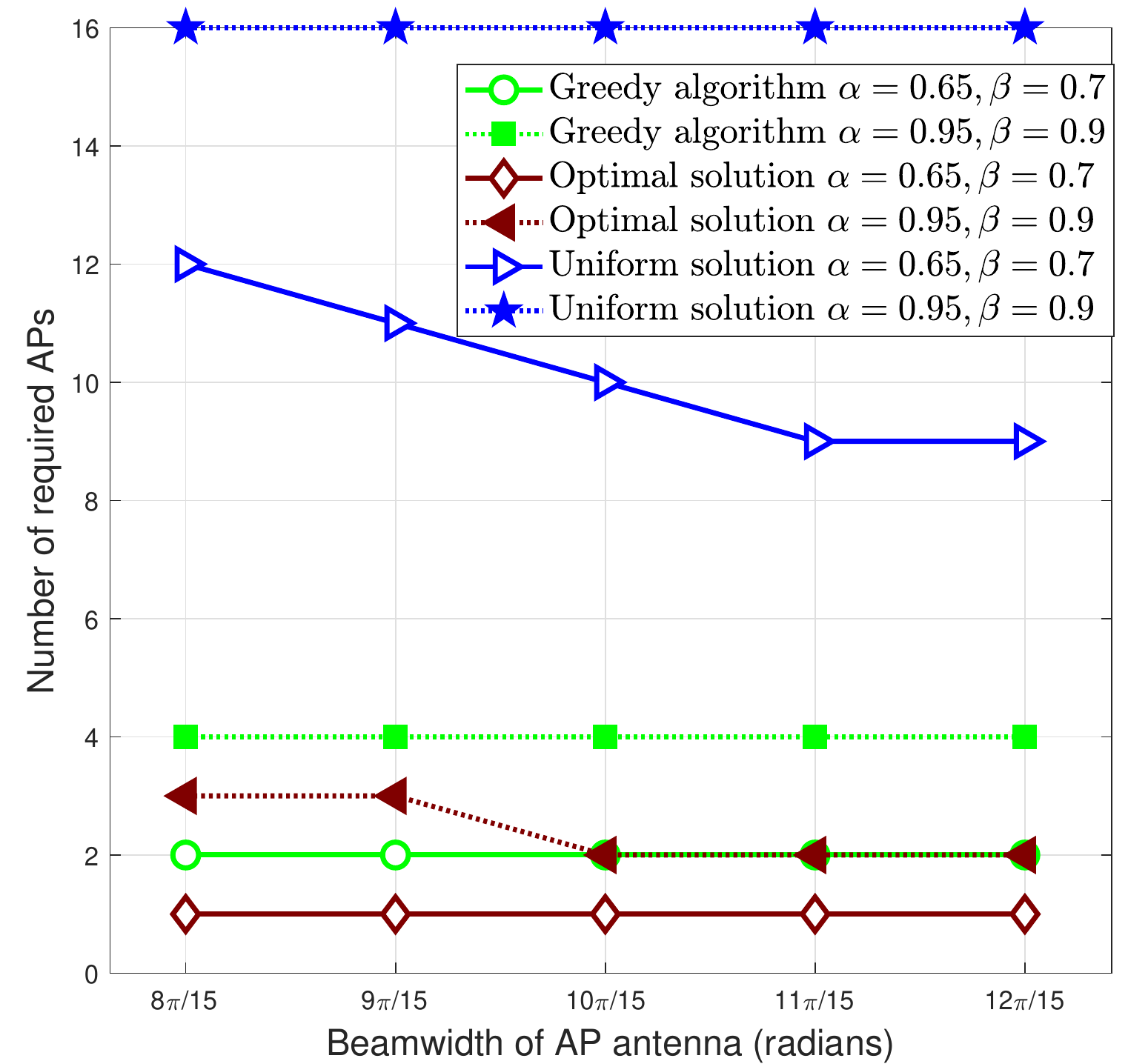}
  \caption{{Number of required APs vs. beamwidth of AP for the meeting room in the Alumni Assembly Hall of Virginia Tech.}}
  \label{N_AP_Hall}
\end{figure}%

In Fig.~\ref{Coverage_Hall}, we show the network coverage versus the coverage constraint. From Fig.~\ref{Coverage_Hall}, we can see that the greedy algorithm guarantees more network coverage compared to the optimal solution. This is an expected result that stems from the fact that more APs are generally deployed by the greedy algorithm than by the optimal solution. When the beamwidth of APs decreases from $\frac{12\pi}{15}$ to $\frac{8\pi}{15}$, the gap of network coverage between the greedy algorithm and the optimal solution increases. From this figure, we can see that, due to the use of additional APs, the greedy algorithm will yield a gain of up to $4\%$ in the network coverage. {Moreover, the network coverage resulting from the greedy algorithm is $13\%$ greater than the uniform solution, in spite of the fact that the greedy solution uses fewer APs than the uniform case (see Fig.~\ref{N_AP_Hall}). This is due to the fact that the coverage of mmW is limited by the possibility of LoS mmW links between the AP and MDs, and the uniform solution does not capture the stochastic changes and blockages in the mmW network.}
\begin{figure}[t!]
  \centering
  \includegraphics[width=.5\linewidth]{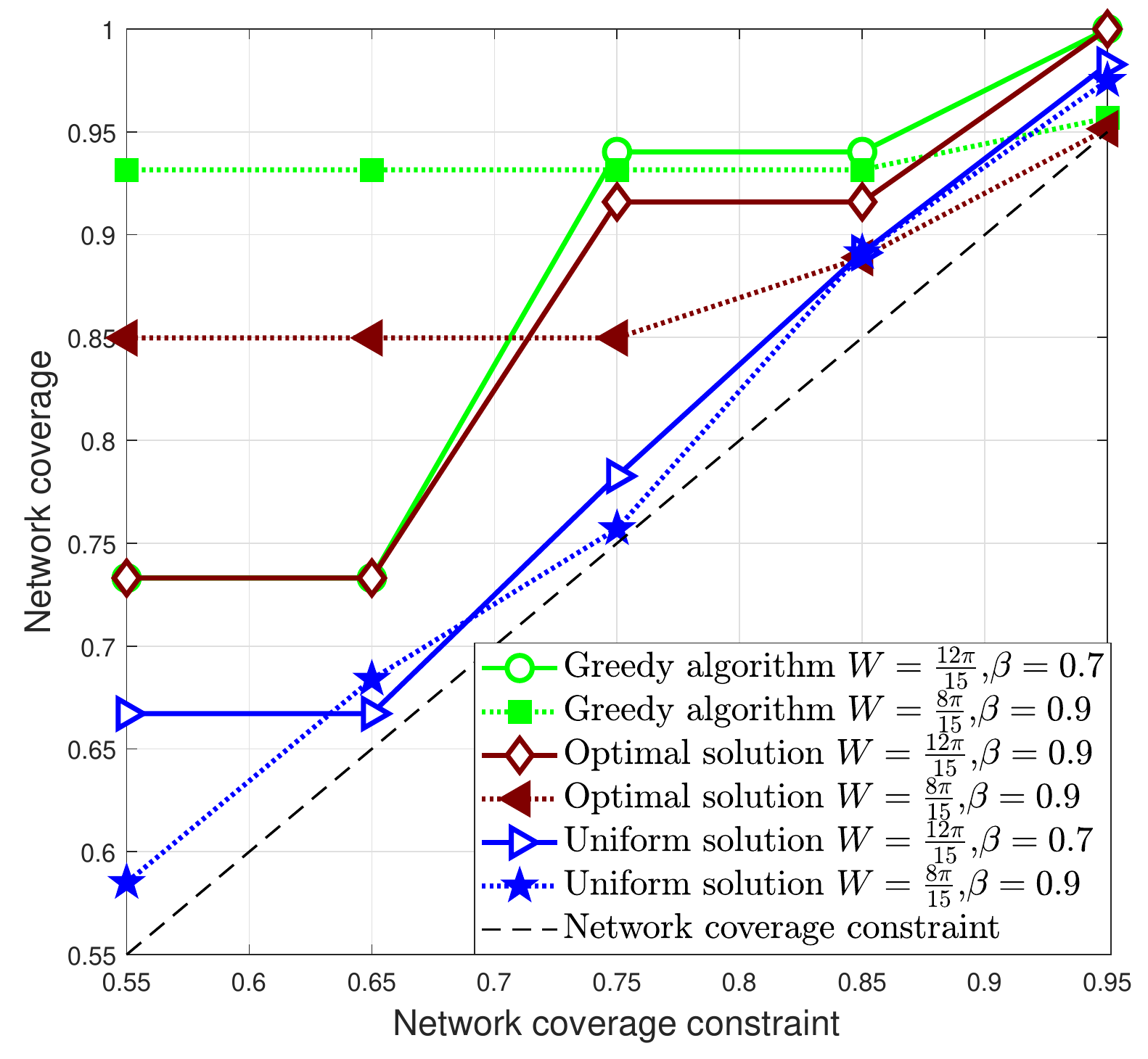}
  \caption{Network coverage vs. network coverage constraint for the meeting room in the Alumni Assembly Hall of Virginia Tech.}
  \label{Coverage_Hall}
\end{figure}

In Fig.~\ref{Band_Hall}, we show the number of AP ratio versus the coverage constraint for different AP beamwidths, when $\beta=0.7$. From Fig.~\ref{Band_Hall}, we can see that the number of AP ratio from simulation results is lower than analytical results for maximum number of AP ratio. By increasing the network coverage constraint, the analytical number of AP ratio also increases. As we can see from Fig.~\ref{Band_Hall}, the difference between analytical and simulation results becomes more pronounced when the AP beamwidth decreases to $W=\frac{8\pi}{15}$. This is due to the fact that, by decreasing the AP beamwidth, the number of required APs increases (see Fig.~\ref{N_AP_Hall}). As such, the additional APs with narrow beamwidth may cover fewer GPs to guarantee the network coverage constraint. Thus, based on Theorem 1, $\frac{\max_{i} C_i^*}{\min_{i} C_i}$ and also the number of AP ratio may increase.
\begin{figure}[t!]
  \centering
  \includegraphics[width=.5\linewidth]{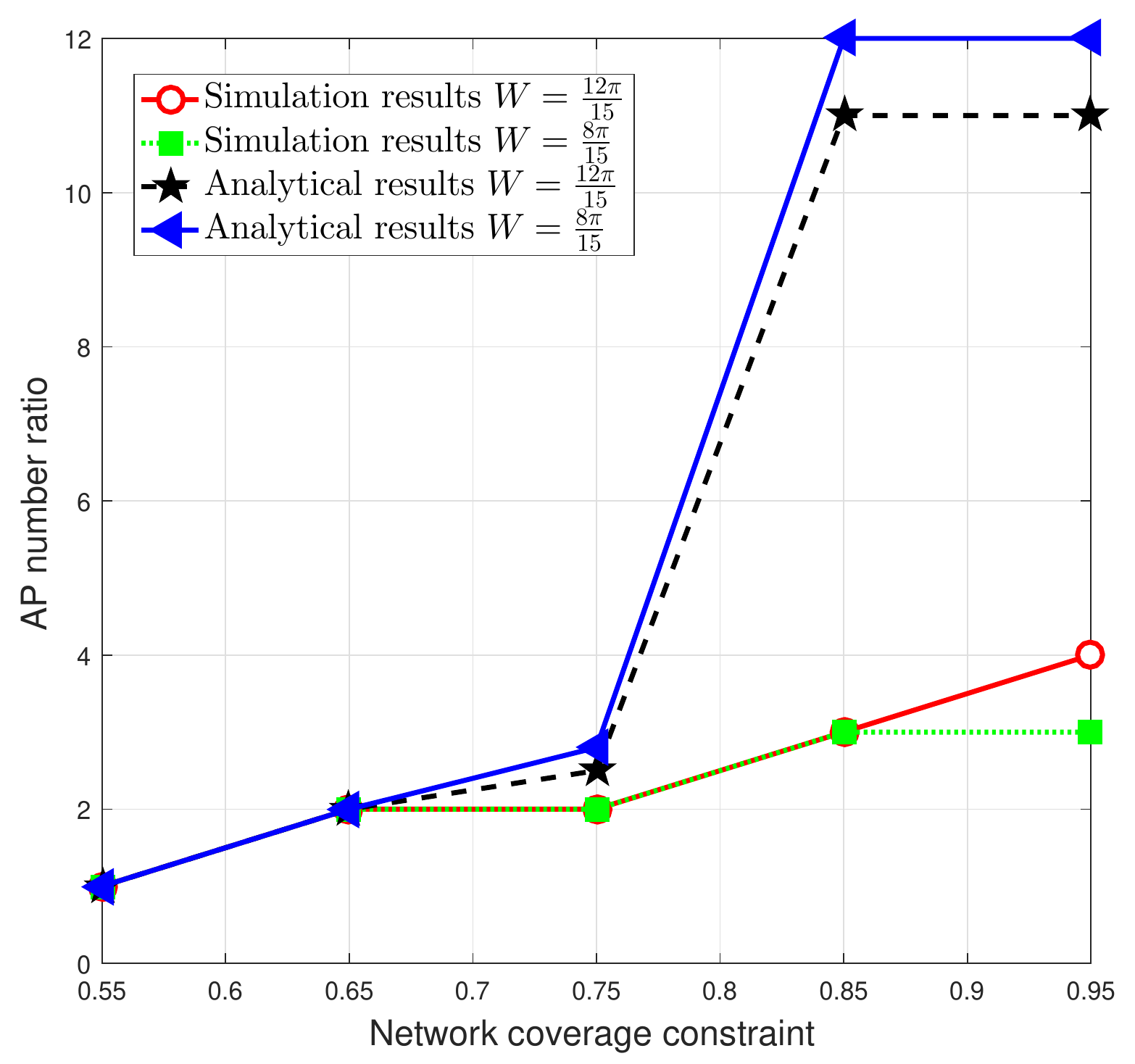}
  \caption{Approximation gap vs. network coverage constraint for the meeting room in the Alumni Assembly Hall of Virginia Tech.}
  \label{Band_Hall}
\end{figure}

In Fig.~\ref{Locdiff_Hall}, we show the AP location difference versus the coverage constraint. From Fig.~\ref{Locdiff_Hall}, we can see that when the network coverage constraint increases, the AP location difference will be more. Fore example, when $\beta=0.7, W=\frac{12\pi}{15}$, and $\alpha=0.55$, the AP location difference is $0$. While the AP location difference is $66\%$ when $\beta=0.7, W=\frac{12\pi}{15}$, and $\alpha$ increases to $0.95$. This is due to the fact that, when the network coverage constraint increases, the greedy solution uses more APs compared to the optimal solution. However, when the network coverage constraint is low and $\beta=0.7$, the greedy and optimal solutions use the same set of APs. Another interesting result from Fig.~\ref{Locdiff_Hall} is that as the user connectively requirement increases, the AP location difference is high. In addition to these, the effect of antenna beamwidth on the AP location difference is negligible.
\begin{figure}[t!]
  \centering
  \includegraphics[width=.5\linewidth]{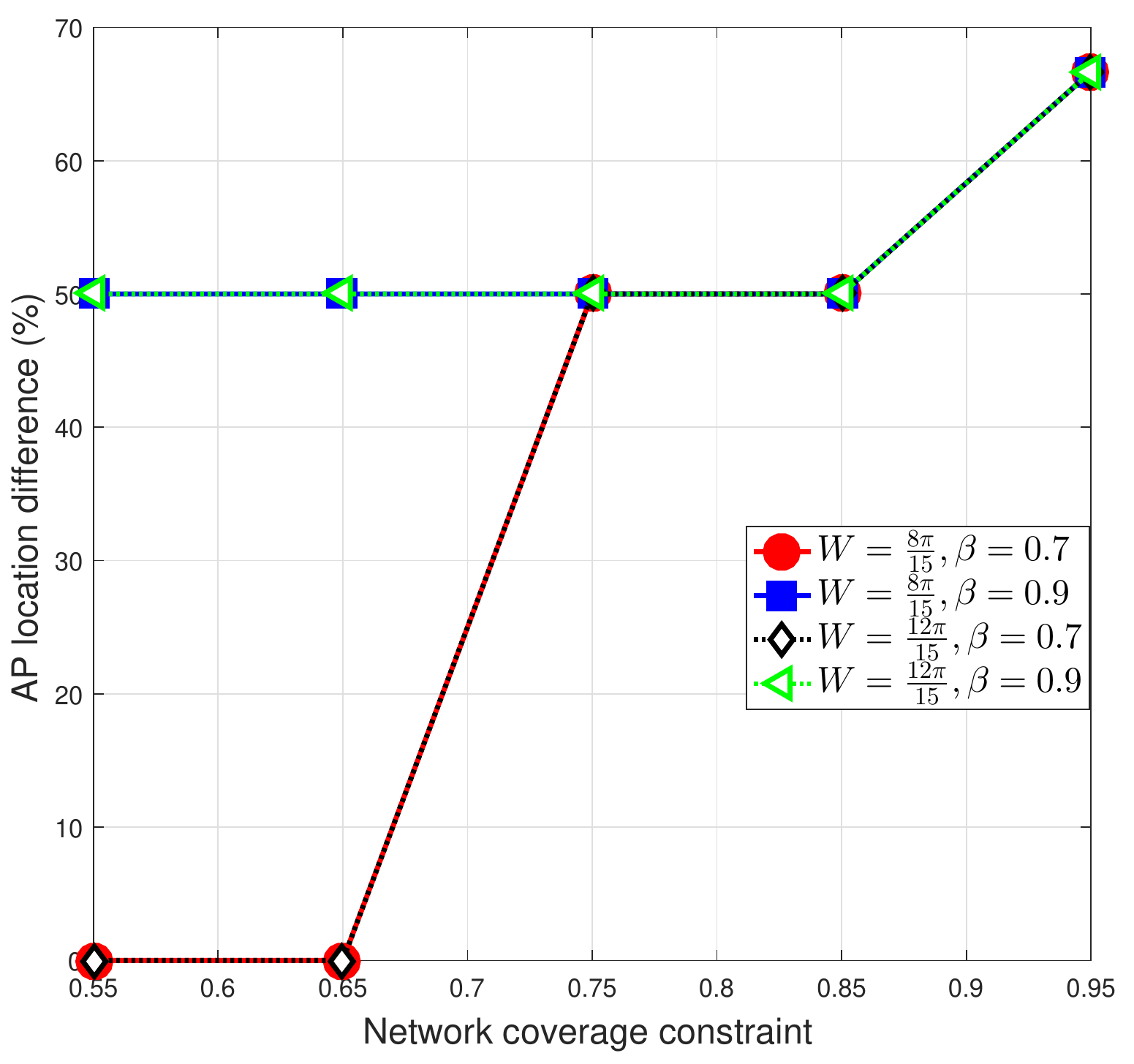}
  \caption{AP location difference vs. network coverage constraint for the meeting room in the Alumni Assembly Hall of Virginia Tech.}
  \label{Locdiff_Hall}
\end{figure}

\subsection{Airport Gate}
The considered airport gate has 160 seats as shown in Fig.~\ref{Airportgate}. The total area of the airport gate is 500 sq. meters with a ceiling height of 10 meters. The number of candidate APs which are on a grid-like structure on the ceiling is 16. In each row of seats, two back-to-back seats are placed while having opposite directions.  Fig.~\ref{Airportgate} shows an illustrative example of AP placement and beam steering resulting from both the optimal solution and the greedy algorithm.
\begin{figure}[t!]
  \centering
  \includegraphics[width=.5\linewidth]{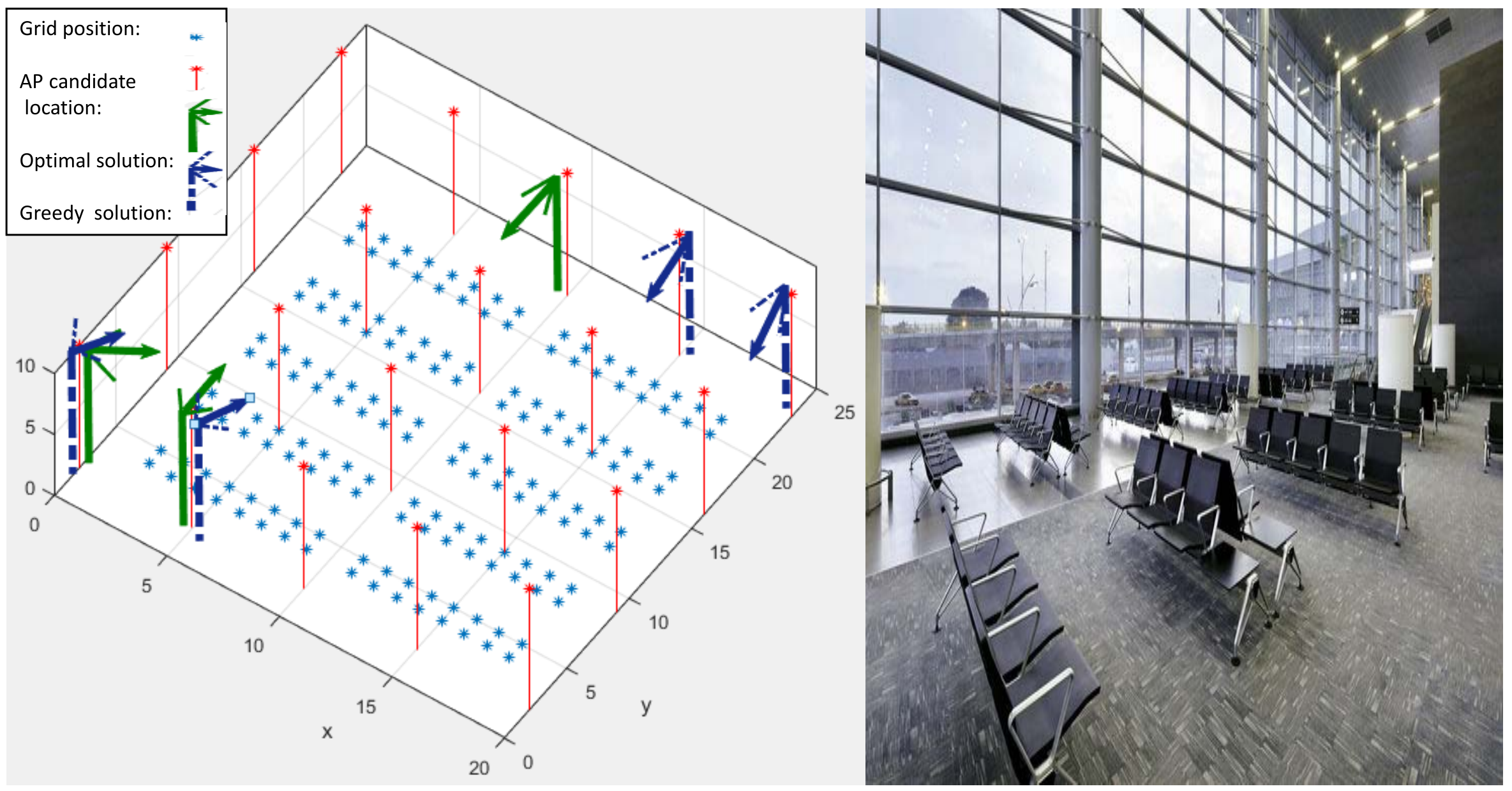}
  \caption{AP deployment and beam steering for an airport gate when $W=\frac{12\pi}{15}$, $w={\pi}$, $\alpha=0.9$, $\beta=0.9$.}
  \label{Airportgate}
\end{figure}

In Fig.~\ref{N_AP_Airportgate}, we show the number of required APs for different network coverage and user connectivity constraints versus the AP beamwidth. Fig.~\ref{N_AP_Airportgate} shows that, when the network coverage and user connectivity constraints increase, the number of required APs also increases. From this figure, we can see that the maximum difference between the number of APs under our proposed greedy solution and optimal one is $2$ for different network coverage when $\beta=0.7$. When the user connectivity requirement, $\beta$, increases from $0.7$ to $0.9$, the greedy solution uses at most 3 more APs compared to the optimal solution for different network coverage.
\begin{figure}[t!]
  \centering
  \includegraphics[width=.5\linewidth]{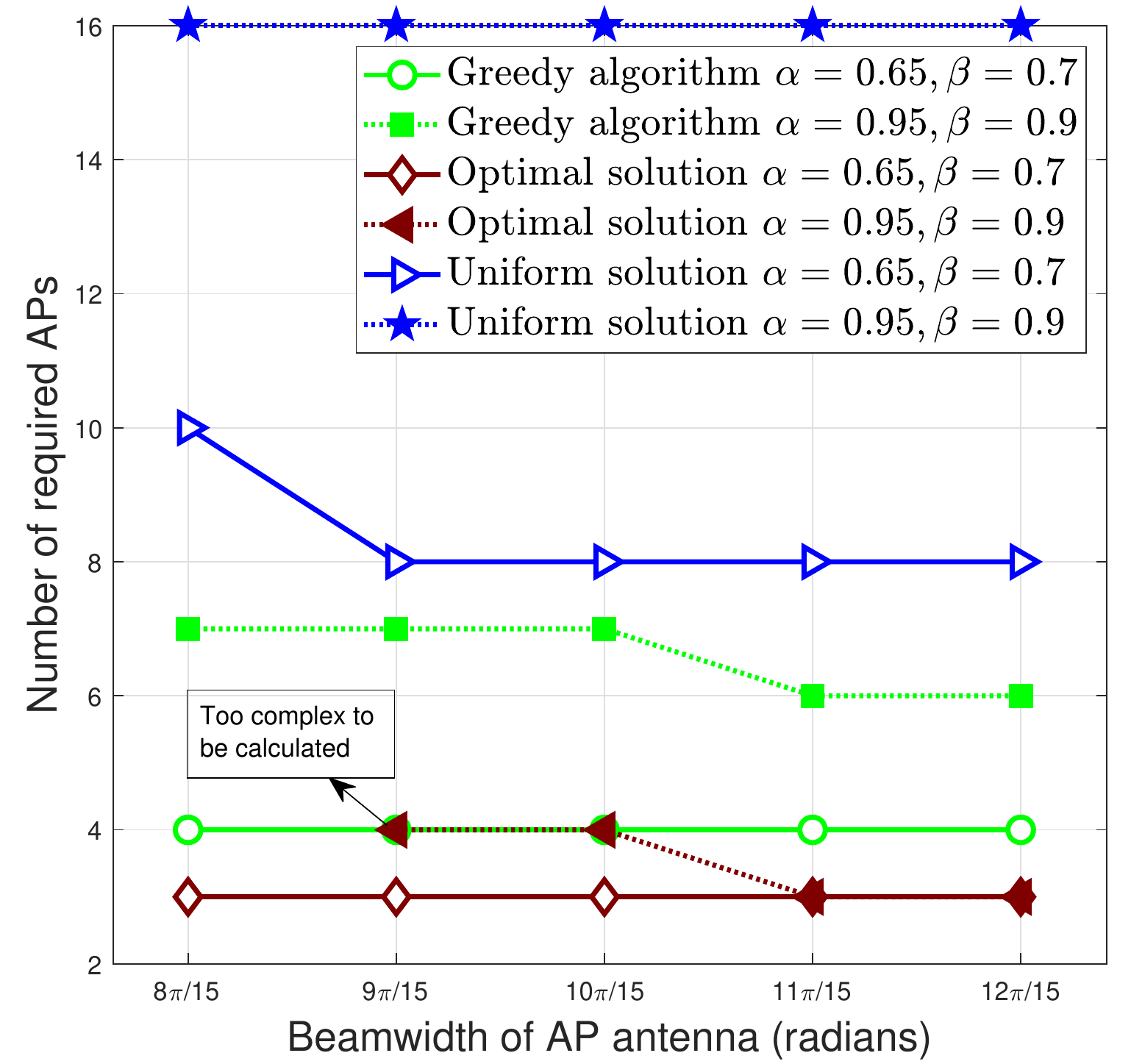}
  \caption{{Number of required APs vs. beamwidth of AP for an airport gate.}}
  \label{N_AP_Airportgate}
\end{figure}%

In Fig.~\ref{Coverage_Airportgate}, we show the network coverage versus the coverage constraint. From Fig.~\ref{Coverage_Airportgate}, we can see that the greedy algorithm guarantees more network coverage compared to the optimal solution because the greedy algorithm uses more APs. Moreover, the optimal solution slightly changes with respect to the beamwidth of APs or user's connectivity constraint. From this figure, we can see that, due to the use of additional APs, the greedy algorithm will yield a gain of up to $11.7\%$ in the network coverage. {Moreover, the uniform solution cannot guarantee the network coverage when the required network coverage increases.}
\begin{figure}[t!]
  \centering
  \includegraphics[width=.5\linewidth]{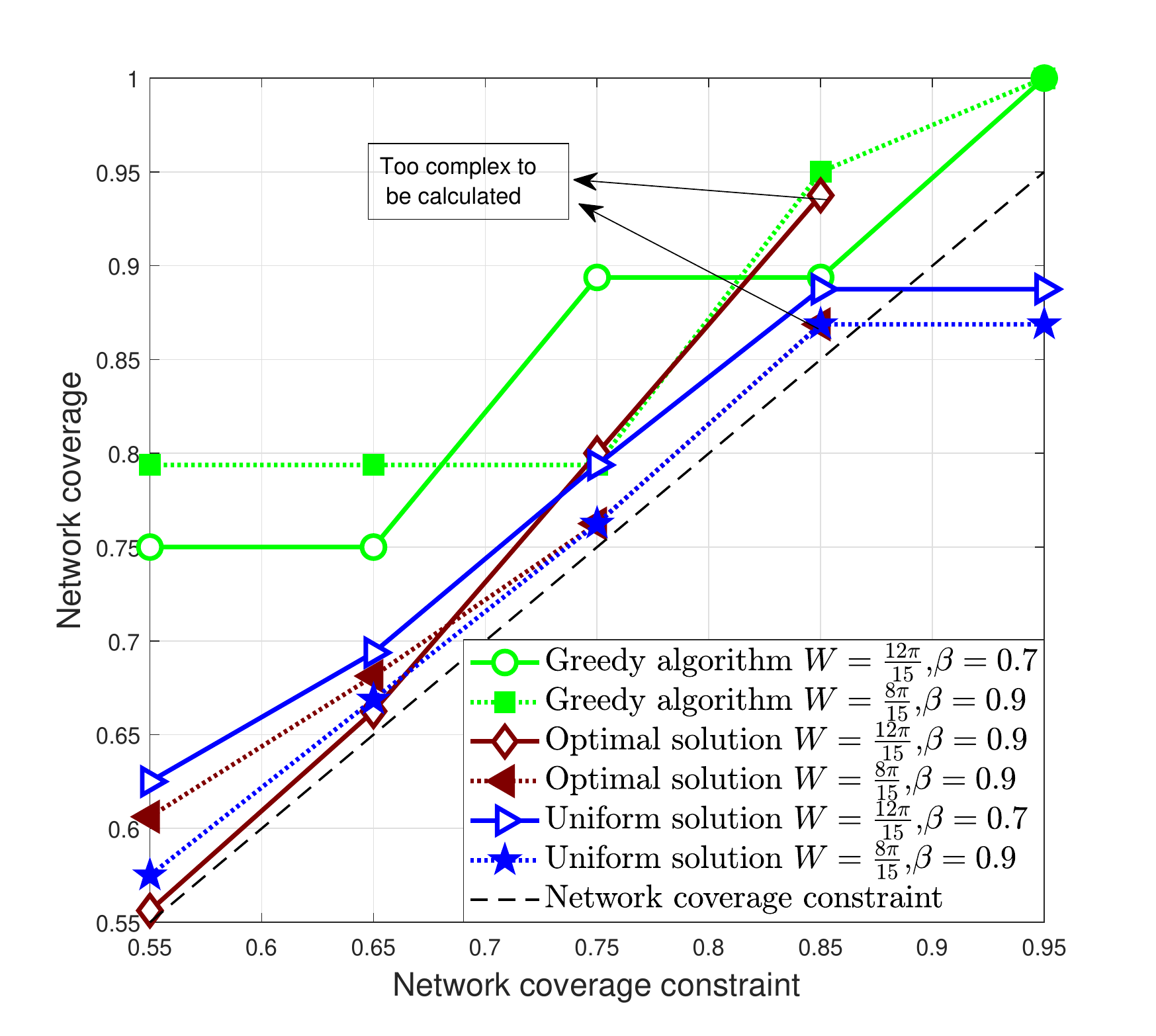}
  \caption{Network coverage vs. network coverage constraint for an airport gate.}
  \label{Coverage_Airportgate}
\end{figure}

In Fig.~\ref{Band_Airportgate}, we show the number of AP ratio versus the coverage constraint for different AP beamwidths, when $\beta=0.7$. From Fig.~\ref{Band_Airportgate}, we can see that the number of AP ratio from simulation results is less than analytical results. Moreover, a more network coverage constraint leads to a higher ratio resulting from the analytical derivations.
\begin{figure}[t!]
  \centering
  \includegraphics[width=.5\linewidth]{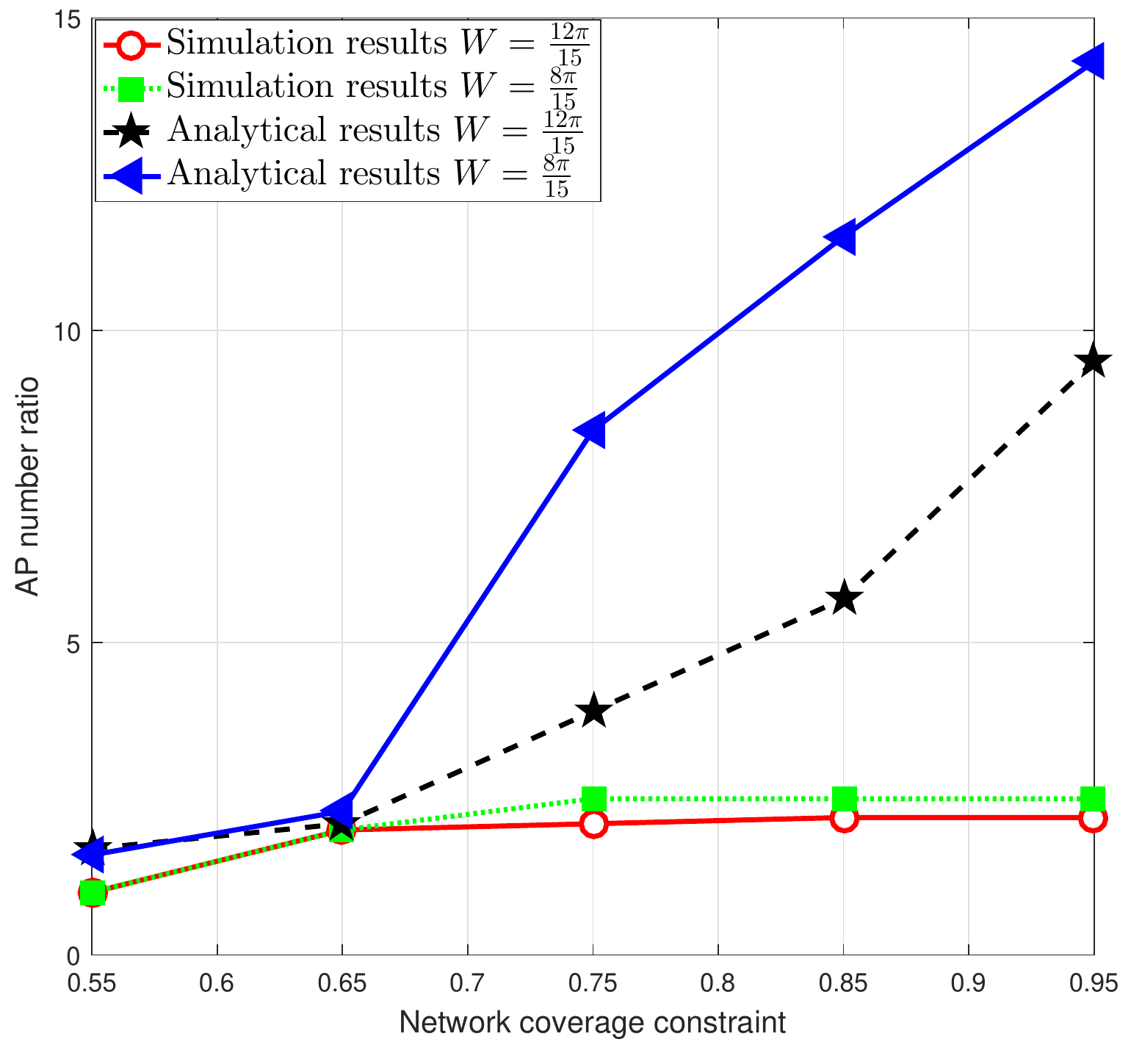}
  \caption{Approximation gap vs. network coverage constraint for an airport gate.}
  \label{Band_Airportgate}
\end{figure}

In Fig.~\ref{Locdiff_Airportgate}, we show the AP location difference versus the coverage constraint. From Fig.~\ref{Locdiff_Airportgate}, we can see that, when the network coverage constraint increases, the AP location difference increases. Fore example, when $\beta=0.7, W=\frac{12\pi}{15}$, and $\alpha=0.55$, the AP location difference is $0$. However, the AP location difference becomes $25$ when $\beta=0.7, W=\frac{12\pi}{15}$, and $\alpha$ increases to $0.95$. This stems from the fact that the greedy solution deploys more APs compared to the optimal solution when the network coverage constraint increases.
\begin{figure}[t!]
  \centering
  \includegraphics[width=.5\linewidth]{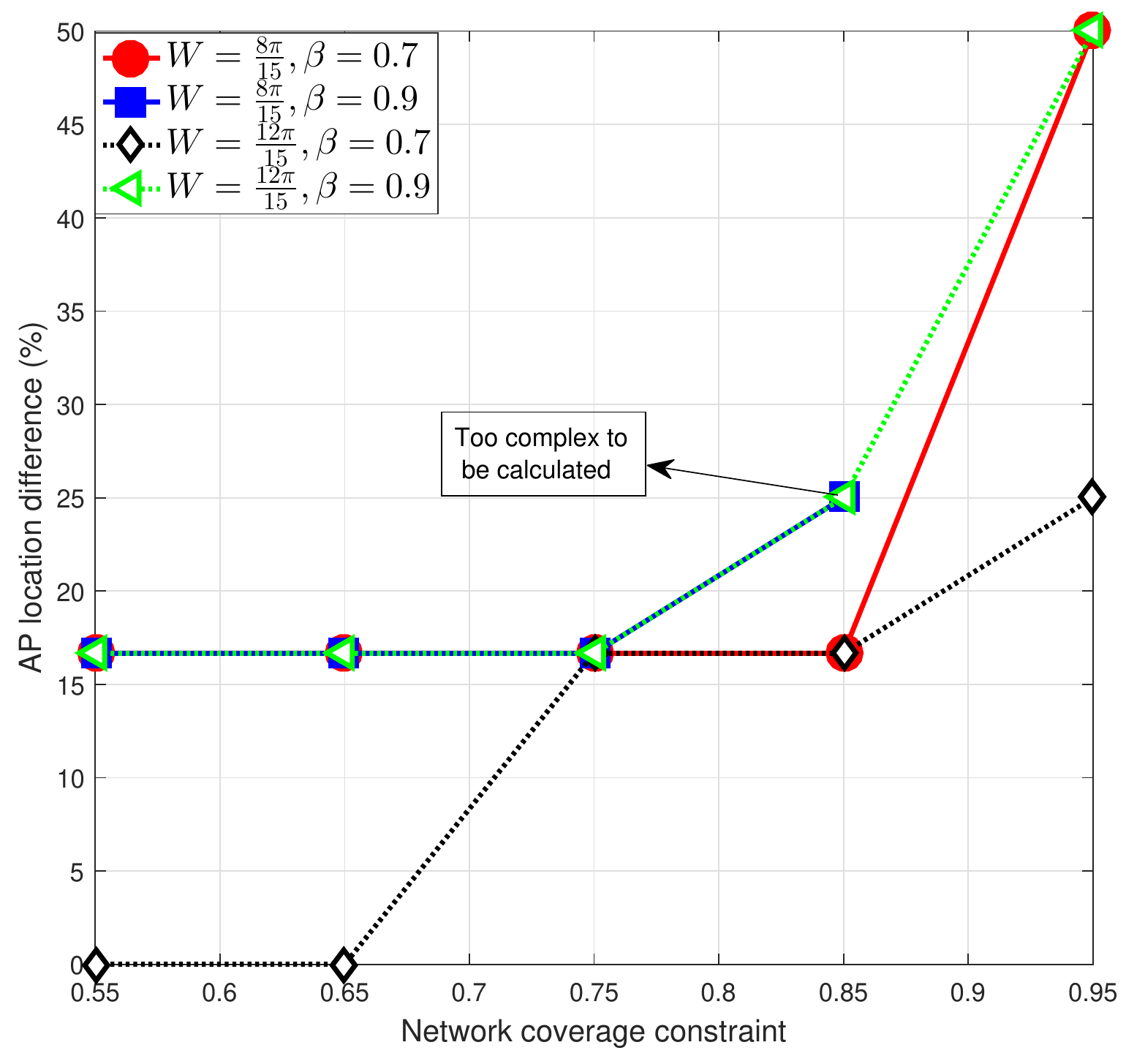}
  \caption{AP location difference vs. network coverage constraint for an airport gate.}
  \label{Locdiff_Airportgate}
\end{figure}

\subsection{Football stadium}
We consider on side of a football stadium with 1040 seats as shown in Fig.~\ref{Footballstadiom}. The height of seats is from 5 to 35 meter. The number of candidate APs which are on a grid-like structure is 16 with the height of the ceiling being 45 meter. The direction of the seats is oriented toward the football field. Fig.~\ref{Footballstadiom} shows an illustrative example of AP placement and beam steering resulting from both the optimal solution and the greedy algorithm. In practice, the UAVs as flying mmW APs can stop at the locations resulting from our proposed algorithm during the game in the open-roof football stadium~\cite{Mozaffari2019}.

\begin{figure}[t!]
  \centering
  \includegraphics[width=.5\linewidth]{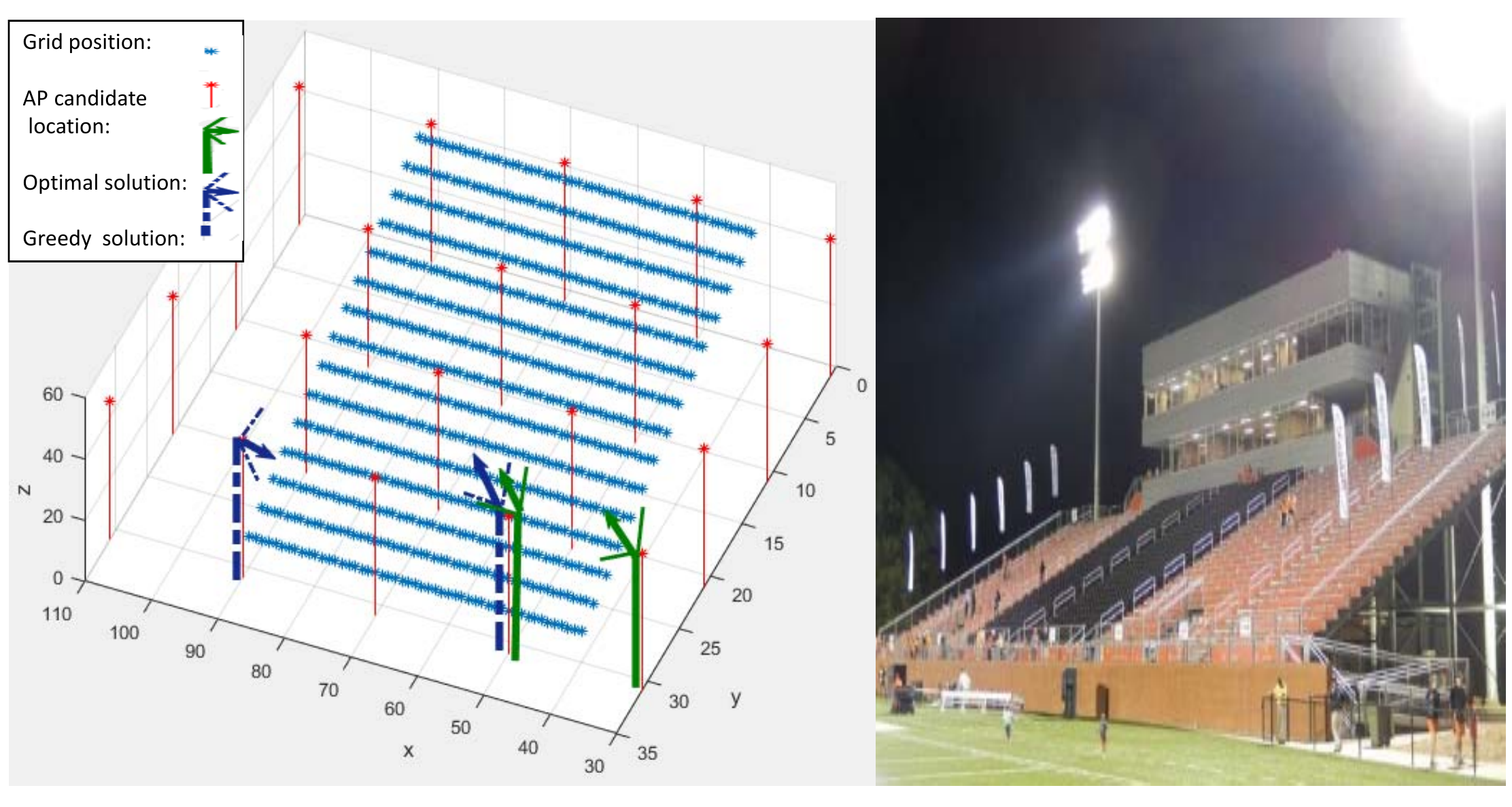}
  \caption{AP deployment and beam steering for a football stadium when $W=\frac{10\pi}{15}$, $w=\frac{\pi}{2}$, $\alpha=0.9$, $\beta=0.7$.}
  \label{Footballstadiom}
\end{figure}%

In Fig.~\ref{N_AP_Footballstadiom}, we show the number of required APs for different network coverage and user connectivity constraints when the AP beamwidth changes. Fig.~\ref{N_AP_Footballstadiom} shows that the more network coverage and user connectivity constraints lead to a higher number of required APs. In addition, for different network coverage and user connectivity constraints, the greedy algorithm uses three additional AP compared to the optimal solution.
\begin{figure}[t!]
  \centering
  \includegraphics[width=.5\linewidth]{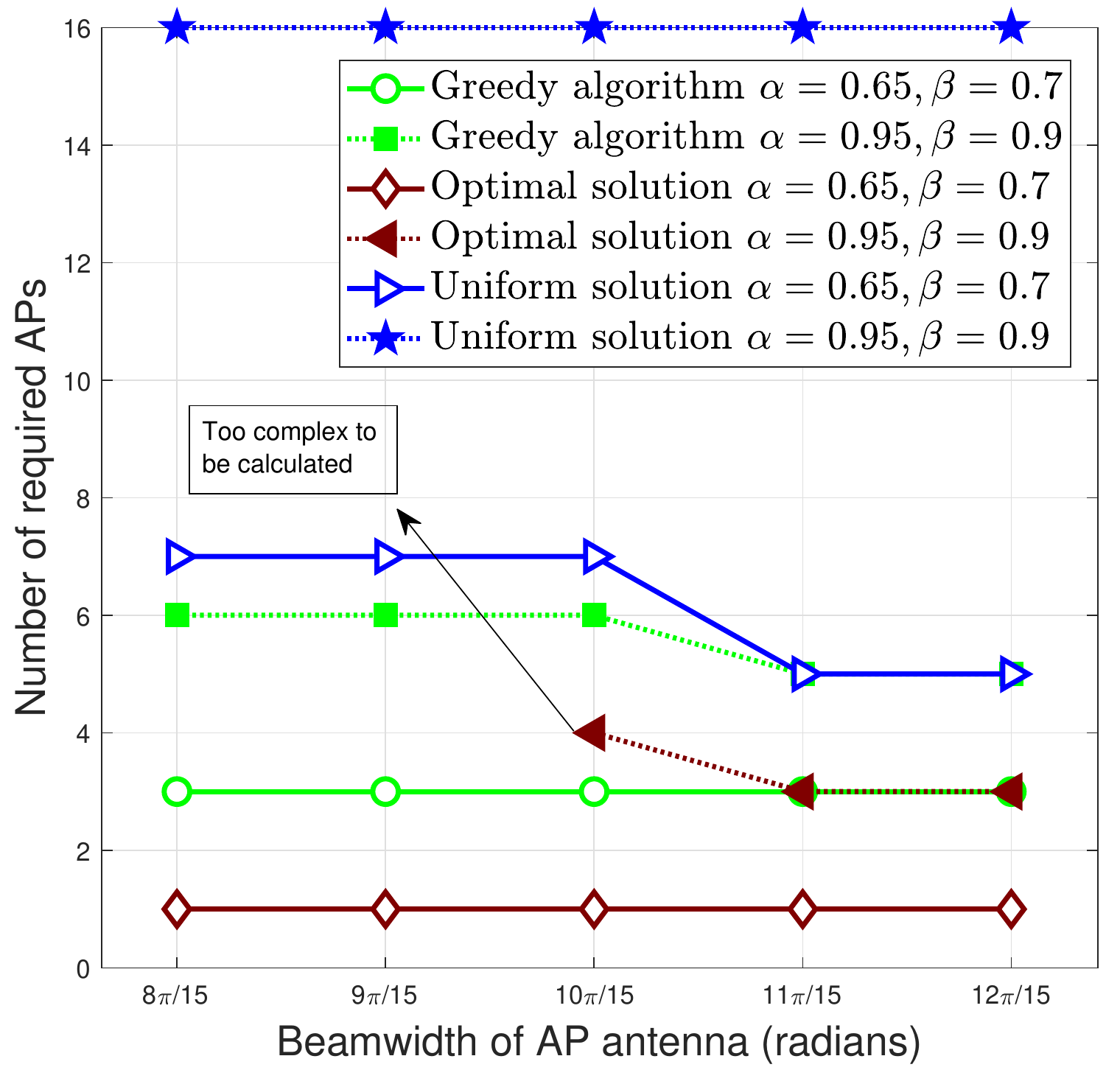}
  \caption{{Number of required APs vs. beamwidth of AP for a football stadium.}}
  \label{N_AP_Footballstadiom}
\end{figure}%

In Fig.~\ref{Coverage_Footballstadiom}, we show the network coverage versus the coverage constraint. From Fig.~\ref{Coverage_Footballstadiom}, we can see that the greedy algorithm guarantees more network coverage compared to the optimal solution because the number of APs from the greedy algorithm are more than the optimal solution (See Fig.~\ref{N_AP_Footballstadiom}). When the beamwidth of APs is small, $\frac{8\pi}{15}$, and the user connectivity constraint is high, $\beta=0.9$, calculating the optimal solution becomes very complex. This is due to the fact that the blockage of nearby users is high for football stadium scenario. From this figure, we can see that, due to the use of additional APs, the greedy algorithm will yield a gain of up to $8\%$ in the network coverage. {Moreover, the network coverage resulting from uniform solution is smaller than the greedy case, although all of the 16 APs are used by the uniform solution (see Fig.~\ref{N_AP_Footballstadiom}).}
\begin{figure}[t!]
  \centering
  \includegraphics[width=.5\linewidth]{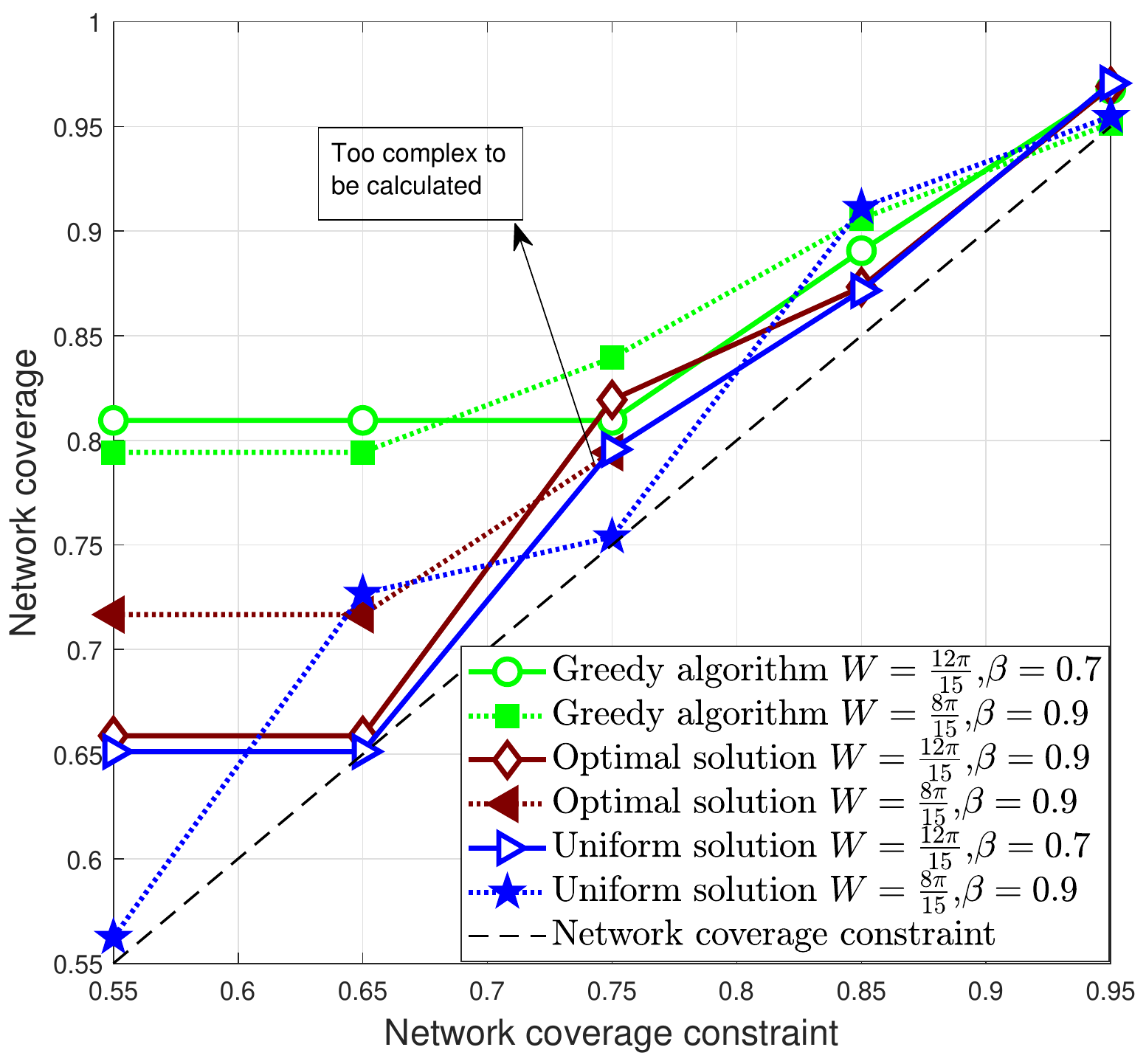}
  \caption{Network coverage vs. network coverage constraint for a football stadium.}
  \label{Coverage_Footballstadiom}
\end{figure}

In Fig.~\ref{Band_Footballstadiom}, we show the number of AP ratio versus the coverage constraint for different AP beamwidth when $\beta=0.7$. From Fig.~\ref{Band_Footballstadiom}, we can see that the ratio derived from analytical results is higher than the one resulting from simulations. By increasing the network coverage constraint the analytical number of AP ratio also increases. As we can see from Fig.~\ref{Band_Footballstadiom}, the difference between analytical and simulation results remains almost same when the beamwidth of AP increases from $W=\frac{8\pi}{15}$ to $W=\frac{12\pi}{15}$.
\begin{figure}[t!]
  \centering
  \includegraphics[width=.5\linewidth]{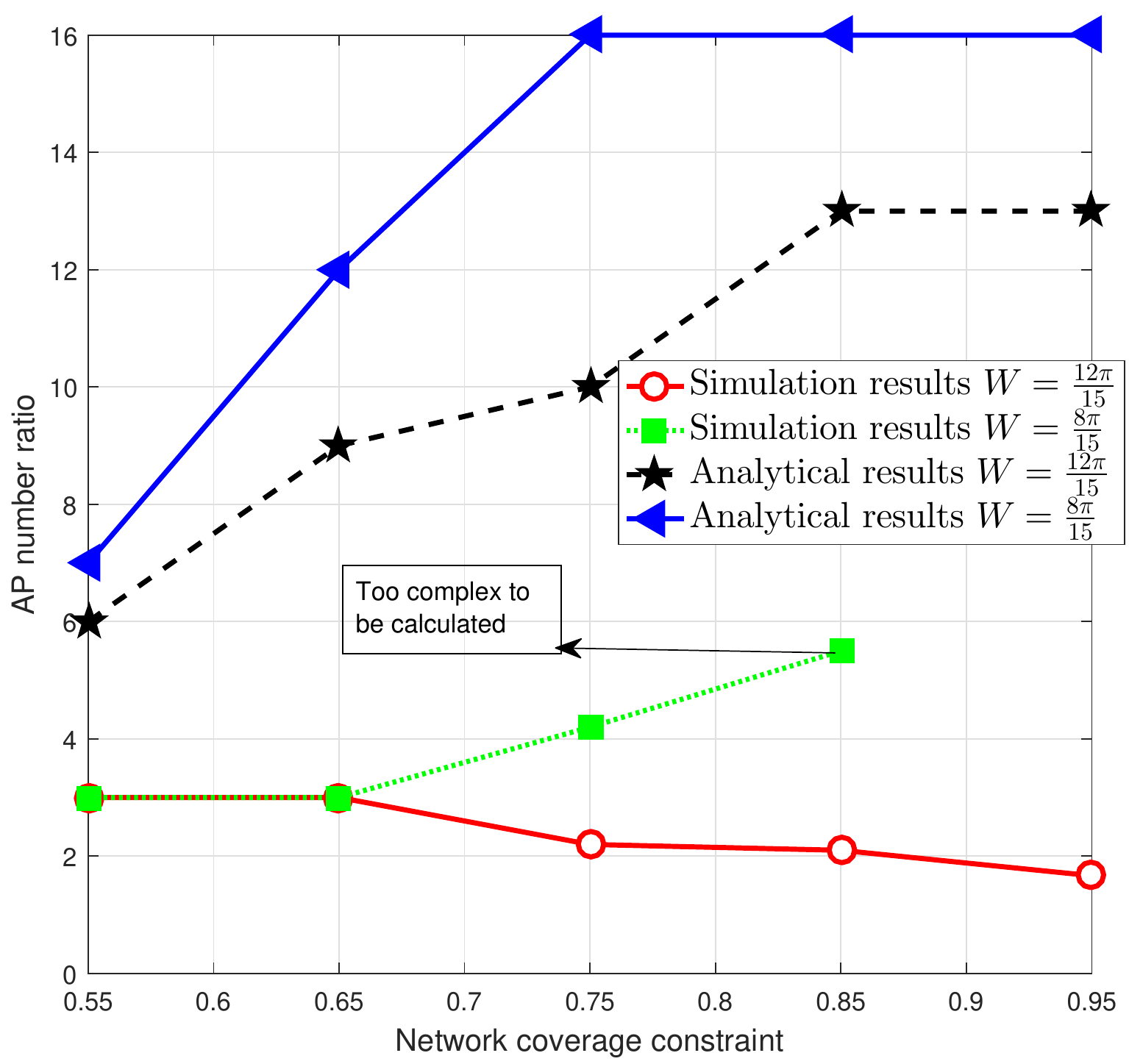}
  \caption{Approximation gap vs. network coverage constraint for a football stadium.}
  \label{Band_Footballstadiom}
\end{figure}

In Fig.~\ref{Locdiff_Footballstadiom}, we show the AP location difference versus the coverage constraint. From Fig.~\ref{Locdiff_Footballstadiom}, we can see that when the network coverage constraint increases, the AP location difference becomes more. Fore example, when $\beta=0.7, W=\frac{12\pi}{15}$, and $\alpha=0.55$ the AP location difference is $0$. But when $\beta=0.7, W=\frac{12\pi}{15}$, and $\alpha=0.95$, the AP location difference becomes $50\%$. This is due to the fact when the network coverage constraint increases, the greedy solution deploys more APs compared to the optimal solution. In addition, the more user connectivity requirement leads to a higher AP location difference. Moreover, the antenna beamwidth does not affect the AP location difference.
\begin{figure}[t!]
  \centering
  \includegraphics[width=.5\linewidth]{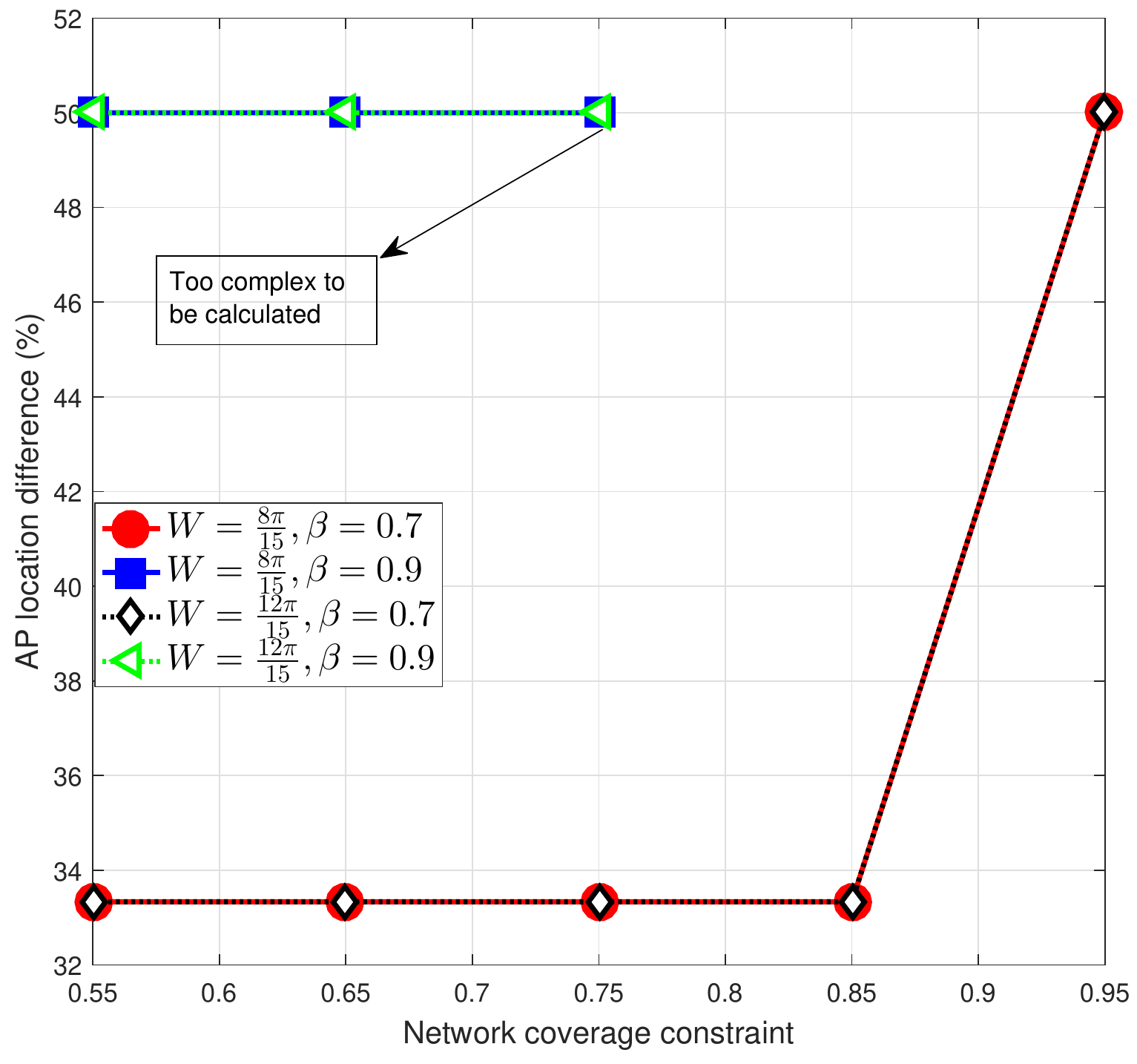}
  \caption{AP location difference vs. network coverage constraint for a football stadium.}
  \label{Locdiff_Footballstadiom}
\end{figure}
\section{Conclusion}\label{Sec:Conclusion}
In this paper, we have studied a joint access point placement and beam steering problem whose goal is to minimize the number of required access points and guarantee coverage for in-venue mmW networks. In the studied model, the availability of LoS mmW links between the APs and the MDs stochastically changes due to the random orientation of the users and the blockage of mmW signals by the users' bodies. First, we have formulated a joint stochastic access point placement and beam steering problem subject to the stochastic user orientation and network coverage constraints. Then, we have designed a greedy algorithm based on size constrained weighted set cover to solve the joint stochastic access point placement and beam steering problem. The approximation ratio between optimal and approximation solutions is derived in closed-form in which the wider beamwidth of APs leads to a smaller approximation gap. Simulation results demonstrate the effectiveness of the proposed approach. For example, the greedy algorithm uses at most 2 additional APs in the Alumni Assembly Hall of Virginia Tech and football stadium, and 3 additional APs in one side of an airport gate compared to the optimal solution in order to guarantee coverage constraint. Moreover, although the greedy algorithm uses additional APs compared to the optimal solution, the greedy algorithm will yield a network coverage that is about $3\%,11.7\%$, and $8\%$ better than the optimal, AP-minimizing solution, for the meeting room in the Alumni Assembly Hall of Virginia Tech, airport gate, and one side of the football stadium, respectively. Another practical gain of our proposed solution is that the complexity of greedy algorithm is much lower than the optimal solution. Due to this fact, solving the optimal solution will be very complex for the scenarios in which the network coverage constraint and user connectivity requirement are high. However, the low-complex greedy algorithm can find the sub-optimal solution much faster than optimal one in the considered scenarios. Future work can extend this approach to cases in which the locations of MDs stochastically change due to the users' mobility, the antenna beamwidth of MD changes based on holding the MD in hand, near head, or in pocket, as well as the optimal deployment of flying mmW access points is required~\cite{7412759}. Another important future work is to take into account the problem of resource allocation, after the network deployment process.

\bibliographystyle{IEEEtran}
\def\baselinestretch{0.9}

\appendix
\section{}
\subsection{Proof of Theorem 1}
For a given set $\mathcal{C} \in \mathcal{M}$, let $u_i(\mathcal{C})=C-\sum_{m\in \mathcal{C}} z_m$ be the number of grid points of $\mathcal{C}$ remaining uncovered in the iteration $i$ of greedy algorithm, where $C$ is the size of set $\mathcal{C}$ and $\mathcal{L}^\circ_{i-1}$ is the set of APs selected by the greedy algorithm until iteration $i-1$. Note that $u_{i-1}(\mathcal{C}) \geq u_i(\mathcal{C})$. So $u_{i-1}(\mathcal{C})-u_i(\mathcal{C})$ is the number of grid points in $\mathcal{C}$ that are covered for the first time in iteration $i$ of the proposed algorithm. Let $\alpha_i(\mathcal{C})=\sum_{m\in \mathcal{C}}q_m z_m$ be the coverage gain of grid points in $\mathcal{C}$ at iteration $i$ of greedy algorithm. We assume that a subset $\mathcal{C}^\circ_i$ of $\mathcal{M}$ selected by the greedy algorithm at iteration $i$ to be assigned to the AP $i$. Let $c_m$ be the price allocated to element $m \in {\mathcal{C}}^\circ_i \cap  \mathcal{M}^*$, that is covered  for the first time at iteration $i$. It is defined as follows:
\begin{equation}
c_m=\frac{
u_{i-1}(\mathcal{C}^\circ_i)- u_i(\mathcal{C}^\circ_i)
}
{
\alpha_i(\mathcal{C}^\circ_i)-\alpha_{i-1}(\mathcal{C}^\circ_i)
}.
\label{price per m}
\end{equation}

For the optimal solution, $\sum_{l\in \mathcal{L}^*}\sum_{m\in \mathcal{C}_l^*}c_m$ is the total price needed to cover GPs in $\mathcal{M}^*$. If some sets in the optimal solution are overlapping, the price of the GPs that are common between those sets will be counted more than once. Hence, we have $\sum_{m\in \mathcal{M}^\circ \cap \mathcal{M}^*}c_m \leq \sum_{l\in \mathcal{L}^*}\sum_{m\in \mathcal{C}^*_l}c_m$. Since $\alpha_i(\mathcal{C}^\circ_i)-\alpha_{i-1}(\mathcal{C}^\circ_i)$ is less than or equal to $\big( u_{i-1}(\mathcal{C}^\circ_i)- u_i(\mathcal{C}^\circ_i) \big) \times \max_{m \in \mathcal{C}^\circ_i} {q_m}$, we can say that $ \frac{1}{\max_{m \in \mathcal{C}^\circ_i} {q_m}} \leq c_m$. Consequently, we can write:
\begin{equation}
\frac{ \sum_{i=1}^{L^\circ} \mathcal{C}_i\cap \mathcal{M}^*  }{\max_{m \in \mathcal{M}^\circ \cap \mathcal{M}^* } {q_m}} \leq \sum_{l\in \mathcal{L}^*}\sum_{m\in \mathcal{C}^*_l\cap \mathcal{M}^\circ}c_m,
\label{calculation1}
\end{equation}
where $\sum_{m\in \mathcal{C}_l^* \cap \mathcal{M}^\circ } c_m =\sum_{i=1}^{L^\circ} \frac {u_{i-1}(\mathcal{C}^*_l\cap \mathcal{M}^\circ)-u_i(\mathcal{C}^*_l\cap \mathcal{M}^\circ)}{\alpha_i(\mathcal{C}^\circ_i)-\alpha_{i-1}(\mathcal{C}^\circ_i)}$. Based on optimization problem in (\ref{Greedprob_Stoc_C}), ${\mathcal{C}}^\circ_i$ is the greedy choice at iteration $i$, so ${\mathcal{C}}^*_l$ cannot increase marginal coverage more than ${\mathcal{C}}^\circ_i$ does. Hence, $\alpha_i(\mathcal{C}^\circ_i)-\alpha_{i-1}(\mathcal{C}^\circ_i)\geq \alpha_i(\mathcal{C}^*_l)-\alpha_{i-1}(\mathcal{C}^*_l)$. Moreover, we can say that $\big(u_{i-1}(\mathcal{C}^*_l\cap \mathcal{M}^\circ)-u_{i}(\mathcal{C}^*_l\cap \mathcal{M}^\circ)\big) \times {\min_{m \in \mathcal{C}^*_l} {q_m}} \leq  \alpha_i(\mathcal{C}^*_l)-\alpha_{i-1}(\mathcal{C}^*_l)$. Thus, we can write:
\begin{align}
&\sum_{m\in \mathcal{C}_l^* \cap \mathcal{M}^\circ }c_m
\leq
\sum_{i=1}^{L^\circ} \frac {u_{i-1}(\mathcal{C}^*_l\cap \mathcal{M}^\circ)-u_i(\mathcal{C}^*_l\cap \mathcal{M}^\circ)}{\big(u_{i-1}(\mathcal{C}^*_l\cap \mathcal{M}^\circ)-u_{i}(\mathcal{C}^*_l\cap \mathcal{M}^\circ)\big) \times {\min_{m \in \mathcal{C}^*_l \cap \mathcal{M}^\circ} {q_m}}}
\label{calculation2}
\end{align}

Considering (\ref{calculation1}) and (\ref{calculation2}), we can write:
\begin{equation}
\frac{ \sum_{i=1}^{L^\circ} \mathcal{C}_i\cap \mathcal{M}^*  }{\max_{m \in \mathcal{M}^\circ \cap \mathcal{M}^* } {q_m}}
\leq
\sum_{l\in \mathcal{L}^*}
\sum_{i=1}^{L^\circ} \frac {u_{i-1}(\mathcal{C}^*_l\cap \mathcal{M}^\circ)-u_i(\mathcal{C}^*_l\cap \mathcal{M}^\circ)}{\big(u_{i-1}(\mathcal{C}^*_l\cap \mathcal{M}^\circ)-u_{i}(\mathcal{C}^*_l\cap \mathcal{M}^\circ)\big) \times {\min_{m \in \mathcal{C}^*_l \cap \mathcal{M}^\circ} {q_m}}}
\label{calculation3}
\end{equation}

If the greedy algorithm covers the grid points that the optimal solution covers, $\mathcal{M}^* \subset \mathcal{M}^\circ$. In this case, $\mathcal{C}^*_l \cap \mathcal{M}^{\circ}=\mathcal{C}^*_l$. Thus, we can write:
\begin{equation}
\frac{ L_*^\circ \times \min_{i \in \mathcal{L^\circ}} C_i^\circ }{\max_{m \in \mathcal{M}^* } {q_m}}\leq
L^* \times \frac{\max_{i \in \mathcal{L^*}} C_i^* }{\min_{m \in \mathcal{M}^*} {q_m}},
\end{equation}
where $L_*^\circ$ is the number of APs selected by greedy algorithm to cover a set of $\mathcal{M}^*$ covered by optimal solution. Considering that $\min_{i \in \mathcal{L}} C_i \leq \min_{i \in \mathcal{L^\circ}} C_i^\circ $, we can say $ L_*^\circ \leq \frac{\max_{i} C_i^* \times \max_{m \in \mathcal{M}^*} {q_m} }{\min_{i} C_i \times \min_{m \in \mathcal{M}^*} {q_m}}L^*$.
\end{document}